\documentclass[11pt]{article}
\usepackage[utf8x]{inputenc}
\usepackage{amsmath, amssymb, amsthm}
\usepackage{graphicx, color, xcolor}
\usepackage[nottoc]{tocbibind}
\usepackage{relsize}
\usepackage{algorithm, algpseudocode}
\usepackage[margin = 1in]{geometry}
\usepackage{url, hyperref}
\usepackage{tikz}
\usetikzlibrary{trees}

\newcommand{\MyComment}[1]{\Comment{\textcolor{blue}{#1}}}
\newcommand{\defeq}{\stackrel{\mathsmaller{\mathsf{def}}}{=}}

\newtheorem{theorem}{Theorem}
\theoremstyle{plain}
\newtheorem{definition}{Definition}
\newtheorem{lemma}{Lemma}
\newtheorem{corollary}{Corollary}

\theoremstyle{remark}
\newtheorem{remark}{Remark}
\newtheorem{obs}{Observation}

\title{Network Size Estimation in Small-World Networks under Byzantine Faults}

\author
{Soumyottam Chatterjee \thanks{Department of Computer Science, University of Houston, Houston, TX 77204, USA. Email: {\tt schatterjee4@uh.edu}.}  \and Gopal Pandurangan \thanks{Department of Computer Science, University of Houston, Houston, TX 77204, USA. Email: {\tt gopalpandurangan@gmail.com}. Research supported, in part, by NSF grant CCF-1527867.}  \and Peter Robinson \thanks{Department of Computing \& Software, McMaster University, Hamilton, Ontario L8S 4L7, Canada. Email: {\tt peter.robinson@mcmaster.ca}.}
}

\begin{document}

\maketitle
\thispagestyle{empty}
\begin{abstract}
We study the fundamental problem of  counting the number of nodes  in a sparse network (of unknown size) under the presence of a large number of Byzantine nodes. We assume the full information model where the Byzantine nodes have complete knowledge about the entire state of the network at every round (including  random choices made by all the nodes), have unbounded computational power, and can deviate arbitrarily from the protocol. Essentially all known algorithms for fundamental Byzantine problems (e.g., agreement, leader election, sampling) studied in the literature assume the knowledge (or at least an estimate) of the size of the network. In particular, all known algorithms for fundamental Byzantine problems that can tolerate a large number of Byzantine nodes in bounded-degree networks assume a {\em sparse expander} network with nodes having knowledge of the network size. It is non-trivial to design algorithms for Byzantine problems that work without knowledge of the network size, especially in bounded-degree (expander) networks where the local views of all nodes are (essentially) the same and limited, and Byzantine nodes can quite easily fake the presence/absence of non-existing nodes. To design truly local algorithms that do not rely on any global knowledge (including network size), estimating the size of the network under Byzantine nodes is an important first step.

Our main contribution is a randomized distributed algorithm that estimates the size of a network under the presence of a large number of Byzantine nodes. In particular, our algorithm estimates the size of a sparse, ``small-world'', expander network with up to $O(n^{1 - \delta})$  Byzantine nodes, where $n$ is the (unknown) network size and $\delta$ can be be any arbitrarily small (but fixed) positive constant. Our algorithm outputs a (fixed) constant factor estimate of $\log(n)$  with high probability; the correct estimate of the network  size will be known to a large fraction ($(1 - \epsilon)$-fraction, for any fixed positive constant $\epsilon$) of the honest nodes. Our algorithm is fully distributed, lightweight, and simple to implement, runs in $O(\log^3{n})$ rounds, and requires nodes to send and receive messages of only small-sized messages per round; any node's local computation cost per round is also small.
\end{abstract}

\newpage

\setcounter{page}{1}


\section{Introduction}\label{sec:intro}

Motivated by the need for robust and secure distributed computation in large-scale (sparse) networks such as peer-to-peer (P2P) and overlay networks, we study the fundamental {\em Byzantine counting} 
problem in  networks, where the goal is to count (or estimate) the number of nodes in a  network that can contain a large number of Byzantine nodes that can exhibit malicious behaviour.

The Byzantine counting problem is challenging  because the goal is to guarantee that most of the honest (i.e., non-Byzantine) nodes obtain a good estimate of the network size 
\footnote{In sparse, bounded-degree networks, an adversary can always isolate some number of honest nodes; hence ``almost-everywhere"  knowledge is the best one can hope for in such networks 
(cf.\ \cite{Dwork_1988}).} 
despite  the presence of a large number of Byzantine nodes  (which have full information about all the nodes and can behave arbitrarily or maliciously) in the network. Byzantine counting is related to, 
yet different, compared with other fundamental problems in distributed computing, namely, {\em Byzantine agreement} and {\em Byzantine leader election}. Similar to the latter two problems, it involves 
solving a global problem under the presence of Byzantine nodes. However, it is a different problem, since protocols for Byzantine agreement or leader election do not necessarily yield a protocol for 
Byzantine counting. In a sense, the Byzantine counting problem can be considered to be more fundamental than Byzantine agreement and leader election, since many existing algorithms for these two problems 
(discussed below and in Section \ref{sec:related}) assume knowledge of the number of nodes in the network $n$; some algorithms require at least a reasonably good estimate of $n$, typically a constant 
factor estimate of $\log n$. Indeed, one of the main motivations for this paper is to design distributed protocols that can work with little or no global knowledge,  including the network size. In this 
sense, an efficient protocol for the Byzantine counting problem can serve as a pre-processing step for protocols for Byzantine agreement, leader election and other problems that either require or assume 
knowledge of an estimate of $n$ \cite{Augustine_2016}.

Byzantine agreement and leader election have been at the forefront of distributed computing research for several decades. The work of Dwork et.\ al.\ \cite{Dwork_1988} and  Upfal \cite{Upfal_1994} 
studied the Byzantine agreement problem in bounded-degree expander networks under the condition of {\em almost-everywhere} agreement, where {\em almost} all (honest) processors need to reach agreement 
as opposed to \emph{all} nodes agreeing as required in the standard Byzantine agreement problem. Dwork et.\ al.\ \cite{Dwork_1988} showed how one can achieve almost-everywhere agreement under up to 
$\Theta(\frac{n}{\log{n}})$  of Byzantine nodes in a bounded-degree {\em expander} network ($n$ is the network size).  Subsequently, Upfal~\cite{Upfal_1994} gave an improved protocol that can tolerate 
up to a linear number of faults in a bounded degree {\em expander} of {\em sufficiently large spectral gap} (in fact, on {\em Ramanujan graphs}, which have the asymptotically largest spectral graph 
possible \cite{Hoory_2006} --- our protocol in this paper also works on a similar type of expander). These algorithms required $O(\log n)$ rounds and polynomial (in $n$) number of messages; however, the 
local computation required by each processor is exponential. Both of the above algorithms  {\em require knowledge of the global topology (including the knowledge of $n$)}, since at the start, nodes need 
to have this information hardcoded. The work of King et.\ al.\ \cite{King_2006_secure_and_scalable_computation} was the first to study scalable (polylogarithmic communication and number of rounds, and 
polylogarithmic computation per processor) algorithms for Byzantine leader election and agreement. Similar to Dwork et al.'s and Upfal's algorithm, the nodes require hardcoded information on the 
network topology --- which is also an {\em expander} graph --- to begin with, including the network size. We note that expansion property is crucially exploited in all the above works to achieve 
Byzantine agreement and leader election. Furthermore, the expander networks assumed in Dwork et al and Upfal works are bounded-degree (essentially, regular) graphs, where without prior knowledge it is 
difficult for nodes to have a knowledge of the network size. 

The works of \cite{Augustine_2012}, \cite{Augustine_2013_agreement}, and \cite{Augustine_2015} studied stable agreement, Byzantine agreement, and Byzantine leader election (respectively) in dynamic 
networks (see also \cite{Augustine_2016}), where in addition to Byzantine nodes there is also adversarial churn. All these works assume that there is an underlying bounded-degree regular expander graph 
and {\em all nodes are assumed to have knowledge of $n$}. It was not clear how to estimate $n$ without additional information under presence of Byzantine nodes in such (essentially, regular and 
constant degree expander) networks. In fact, the works of \cite{Augustine_2016, Augustine_2015} raised the question of designing protocols in expander networks that work when the network size is not 
known and may even change over time, with the goal of obtaining a  protocol that works when nodes have strictly local knowledge. This requires devising a distributed protocol that can measure global 
network parameters such as size, diameter, average degree, etc.\ under Byzantine nodes in sparse networks, especially in sparse expander networks.

\subsection{Our Contributions}

We introduce and study the problem of Byzantine Counting. Our goal is to design a distributed algorithm that guarantees, despite a large number of Byzantine nodes, that almost all honest nodes know a good estimate of the network size in a bounded degree, ``small world'' network. We are not aware of any prior work that studies {\em Byzantine} counting in the setting addressed here.

Before stating our result, we briefly describe the key ingredients of our network model (we refer to Section \ref{sec:model} for the full details). We assume a sparse network that has constant 
bounded degree (essentially regular) and has {\em high  expansion} as well as large {\em clustering coefficient}. In other words, it is a ``small-world'' network. Expander graphs have 
been used extensively as candidates to solve the Byzantine agreement and related problems in bounded degree graphs (e.g., as discussed earlier, see \cite{Dwork_1988, Kapron_2010, King_2011, 
King_2006_scalable_leader_election, Upfal_1994}); the expander property proves crucial in tolerating a large number of Byzantine nodes. The high expansion of such graphs have been exploited in previous works as well, most notably by Upfal \cite{Upfal_1994} to solve the Byzantine agreement (with knowledge of $n$). For the Byzantine counting problem, which seems harder, however, expansion by itself does not seem to be sufficient; our protocol also exploits the high clustering coefficient of the network crucially (cf. Section \ref{sec:technical}).

We assume that up to $O(n^{1 - \delta})$ nodes can be \emph{Byzantine}, where $\delta > 0$ can be an arbitrarily small (but fixed) positive constant. We assume a strong adversarial {\em full information model} where the Byzantine nodes (who have unbounded computational power) are {\em adaptive}, in the sense that they know the entire states of all nodes at the beginning of every round (including the messages sent by them), including the random choices made by the nodes up to and including the current round as well as future rounds (in other words, they are {\em omniscient}). However, we note that the Byzantine nodes can communicate only using the edges of the network, i.e., they can send messages directly only to their neighbors.

In our network model, where nodes have constant bounded degree, most nodes, with high probability, see (essentially) the same local topological structure even for a reasonably large neighborhood radius ---  cf.\ Section \ref{section-locally-tree-like-property}, and hence nodes do not have any a priori local information that can help them estimate the network size. In this setting, Byzantine nodes can easily fake the presence/absence of nodes --- thus trying to foil the estimate of the honest nodes.

Our main contribution is a distributed algorithm (cf. Section \ref{section-algorithm}) that estimates the size of the network, even  under the presence of a large number of Byzantine nodes. In 
particular, our algorithm estimates the size of a sparse (constant  degree) ``small-world'' network with up to $O(n^{1 - \delta})$  (for any small positive constant $\delta$) Byzantine nodes,  where $n$ is the (unknown) network size. Our algorithm outputs a (fixed) constant factor estimate of $\log{n}$  with high probability; 
\footnote{``With high probability (whp)" refers to a probability  $\ge 1 - n^{-c}$, for some constant $c > 0$.} 
the correct estimate of the network size will be known to  $(1 - \epsilon)$-fraction  (where $\epsilon > 0$ is a small constant, say 0.1) of the honest nodes.  
\footnote{We call $\epsilon$ \emph{the error parameter} --- by changing its value (please refer to Line \ref{value-of-alpha-i} in the pseudocode in Algorithm \ref{alg}), we (the algorithm designer) can control exactly how large a fraction of the honest nodes would estimate $\log{n}$ correctly (i.e., get a constant-factor approximation of $\log{n}$). Theorem \ref{theorem-main-result-of-the-paper}, which is the main result of this paper, tells us that at most $\epsilon$-fraction of the honest nodes would \emph{fail} to get a constant factor approximation of $\log{n}$.}

Our algorithm is the first known, decentralized Byzantine counting algorithm that can tolerate a large amount of Byzantine nodes. It is fully-distributed, localized (does not require any global 
topological knowledge), lightweight, runs in $O(\log^3 n)$ rounds, and requires nodes to send and receive ``small-sized messages'' only. 
\footnote{A ``small-sized message'' is one  that contains a constant number of IDs and $O(\log{n})$ additional bits.}
Any node's computation cost per round is also logarithmic.

The given algorithm is a basic ingredient that can be used for the design of efficient distributed algorithms resilient against Byzantine failures, where the knowledge of the network size (a global parameter) may  not be known a priori. It can serve as a building block for implementing other non-trivial distributed computation tasks in Byzantine networks such as agreement and leader election where the network size (or its estimate) is not known a priori.
\subsection{Technical Challenges}\label{sec:technical}

The main technical challenge that we have to overcome is designing and analyzing distributed algorithms under the presence of Byzantine nodes in networks where (honest) nodes only have local knowledge, 
i.e., knowledge of their immediate neighborhood.  It is possible to solve the counting problem exactly in networks without Byzantine nodes by simply building a spanning tree and converge-casting the 
nodes' counts to the root, which in turn can compute the total number of nodes in the network. A more robust and alternate way that works also in the case of {\em anonymous} networks is the technique 
of {\em support estimation} \cite{Augustine_2012, Augustine_2016} which uses {\em exponential} distribution (or alternately one can use a {\em geometric} distribution, see e.g., \cite{Kutten_2015}) to 
estimate accurately the network size as described below.

Consider the following simple protocol for estimating the network size that uses the geometric distribution. Each node $u$ flips an unbiased coin until the outcome is \texttt{heads}; let $X_u$ denote 
the random variable that denotes the number of times that $u$ needs to flip its coin. Then, nodes exchange their respective values of $X_u$ whereas each node only forwards the highest value of $X_u$ 
(once) that it has seen so far. We observe that $X_u$ is geometrically distributed and denote its global maximum by $\bar{X}$. For any $u$,  $\Pr(X_u \geq 2\log n)  =  (\frac{1}{2})^{2\log{n}}$, and by 
taking a union bound, $\Pr(\bar{X} \geq 2\log{n})  \leq  \frac{1}{n}$. Furthermore, $\Pr(\bar{X} < \frac{\log{n}}{2})  =  (1 - (\frac{1}{2})^{\frac{\log{n}}{2}})^n  \leq  e^{-\sqrt{n}}$. It follows that 
each node forwards at most $O(\log{n})$ distinct values (w.h.p.). After $O(D)$ rounds (where $D$ is the network diameter), each node knows the value of $\bar{X}$,  and sets that as its estimate of 
$\log{n}$. Due to the above bounds on $\bar{X}$ it follows that (w.h.p.), it is a constant factor estimate of $\log{n}$. The support estimation algorithm \cite{Augustine_2012, Augustine_2016} which uses 
the exponential distribution works in a  similar manner.

The geometric distribution protocol fails when even just one Byzantine node is present. Byzantine nodes can fake the maximum value or can stop the correct maximum value from spreading and hence can 
violate any desired approximation guarantee. Hence a new protocol is needed when dealing with Byzantine nodes.

Prior localized techniques that have been used successfully for solving other problems such as Byzantine agreement and leader election  such as random walks  and majority agreement  (e.g., 
\cite{Augustine_2013_agreement, Augustine_2015}) do not imply efficient (i.e., fast algorithms that uses small message sizes) algorithms for Byzantine counting. For instance, random walk-based techniques 
crucially exploit a uniform sampling of tokens (generated by nodes) after $\Theta(\text{mixing time})$ number of steps. However, the main difficulty in this approach is that the mixing time is unknown 
(since the network size is unknown) --- and hence it is unclear a priori how many random walk steps the tokens should take. Similar approaches based on the return time of random walks fail due to long 
random walks having a high chance of encountering a Byzantine node. One can also use Birthday paradox ideas to try to estimate $n$ (e.g., these have been tried in an non-Byzantine setting 
\cite{Ganesh_2007}); these also fail in the Byzantine case.

We note that one can possibly solve Byzantine counting if one can solve Byzantine leader election;
\footnote{Informally, the idea is as follows. If one can elect a honest leader, then it can initiate flooding by sending a message to the entire network; any other  node can set an estimate of $\log{n}$ 
as the round number when it sees  the message for the first time. It can be shown that in a sparse expander, $n - o(n)$ nodes will have a constant factor estimate of $\log n$.}
however, all known algorithms for Byzantine leader election (or  agreement) {\em assume a priori knowledge (or at least a good estimate) of the network size}.  Hence we require a new protocol that 
solves Byzantine counting from ``scratch.''  In our random network model, where most nodes, with high probability, see  (essentially) the same local topological structure (and constant degree) even for 
a reasonably large neighborhood radius (cf.\ Section \ref{section-locally-tree-like-property}), it is difficult for nodes to break symmetry or gain a priori knowledge of $n$.
\footnote{We point out that with constant probability, in our network model, due to the property of the $d$-regular random graph, an expected constant number of nodes might have multi-edges --- this can 
potentially be used to break ties; however, this fails to work with constant probability. In any case, such symmetry breaking will fail in symmetric regular graphs.}

Another approach is to try to estimate the diameter of the network, which, being $\Theta(\log{n})$ for sparse expanders, can be used to deduce an approximation of the network size. Assuming that there 
exists a leader in the network, one way to do this is for the leader to initiate the flooding of a message and it can be shown that a  large fraction of nodes (say a $(1 - \epsilon)$-fraction, for some 
small $\epsilon > 0$) can estimate the diameter by recording the time when they see the first token, since we assume a synchronous network. However, this method fails since it is not clear, how to break 
symmetry initially by choosing a leader --- this by itself appears to be a  hard problem in the Byzantine setting without knowledge of $n$.

We now give a high-level intuition behind our protocol. The main idea is based on using the geometric distribution, but there are several technical obstacles that we need to tackle (cf. Section 
\ref{section-algorithm}).

The algorithm operates in phases. In phase $i$, each honest node estimates the number of nodes at distance $i$ (in particular, whether there are any nodes at all) by observing the maximum (or 
near-maximum) value, generated according to the geometric distribution, at distance $i$; this value can be propagated by flooding for exactly $i$ steps. We only allow certain values to propagate in 
phase $i$; this {\em avoids congestion} and hence our algorithm works using only small message sizes. As $i$ increases, i.e., when it becomes $a\log n$, for some small constant $0 < a < 1$, this provides 
a constant factor estimate of $\log{n}$. Up to a distance of $i = a\log{n}$, most nodes (i.e., $n - o(n)$ nodes) do not see any values from Byzantine nodes, since most nodes are a distance at least 
$a\log n$ from any Byzantine node --- this is due to the property of the {\em expander} graph. However, as $i$ increases, the Byzantine nodes can introduce fake values and hence can fool most of the 
nodes into believing that the network is much larger than it actually is. To overcome this, the protocol  exploits the {\em small-world} property of the network, i.e., nodes have high clustering 
coefficients --- which implies that a node's neighbors are well-connected among themselves. Each (honest) node checks with its neighbors to see if the value sent by the Byzantine node is consistent among 
the neighbor set; if not, this (high) value is discarded.

There are some complications in implementing this idea, since Byzantine nodes can lie about the identity of neighboring nodes; our protocol exploits the fact that the network is a union of expander and 
small-world network to overcome this. We refer to Section \ref{section-algorithm-with-byz} for more details.
\subsection{Other Related Works}\label{sec:related}

There have been several works on estimating the size of the network, see e.g., the works of \cite{Ganesh_2007,Horowitz_2003, Luna_2014, Terelius_2012, Shafaat_2008}, but all these works do not work under 
the presence of Byzantine adversaries. There have been some work on using network coding for designing byzantine protocols (see e.g., \cite{Jaggi_2008}); but these protocols have polynomial message sizes 
and are highly inefficient for problems such as counting, where the output size is small. There are also some works on topology discovery problems under Byzantine setting (e.g., \cite{Nesterenko_2006}), 
but these do not solve the counting problem.

Several recent works deal with Byzantine agreement, Byzantine leader election, and fault-tolerant protocols in dynamic networks.  We refer to \cite{Guerraoui_2013, Augustine_2012, 
Augustine_2013_agreement, Augustine_2013_storage_and_search, Augustine_2015} and the references therein for details on these works. These works crucially assume the knowledge of the network size (or at 
least an estimate of it) and don't work if the network size is not known.

There have been significant work in designing peer-to-peer networks that are provably robust to a large number of Byzantine faults~\cite{Fiat_2002, Hildrum_2003, Naor_2003, Scheideler_2005}. These focus 
only on (robustly) enabling storing and retrieving data items. The problem of achieving almost-everywhere agreement among nodes in P2P networks (modeled as expander graphs) is considered by King et al.\ 
in \cite{King_2006_secure_and_scalable_computation} in the context of the leader election problem; essentially, \cite{King_2006_secure_and_scalable_computation} is a sparse (expander) network 
implementation of the full information protocol of \cite{King_2006_scalable_leader_election}. In another recent work \cite{King_2014}, the authors use a spectral technique to ``blacklist'' malicious 
nodes leading to faster and more efficient Byzantine agreement. The work of \cite{Guerraoui_2013} presents a solution for maintaining a clustering of the network, where each cluster contains more than 
two-thirds honest nodes with high probability in a setting where the size of the network can vary polynomially over time. All the above works assume an exact knowledge of or some good estimate of the 
network size and do not solve the Byzantine counting problem.

The work of \cite{Bortnikov_2009} shows how to implement uniform sampling in a peer-to-peer system under the presence of Byzantine nodes where each node maintains a local ``view'' of the active nodes. We 
point out that the choice of the view size and the sample list size of $\Theta(n^{\frac{1}{3}})$ necessary for withstanding adversarial attacks requires the nodes to have a priori knowledge of a 
polynomial estimate of the network size. \cite{Horowitz_2003} considers a dynamically changing network \emph{without} Byzantine nodes where nodes can join and leave over time and provides a local 
distributed protocol that achieves a polynomial estimate of the network size. In \cite{Bovenkamp_2012}, the authors present a gossip-based algorithm for computing aggregate values in large dynamic 
networks (but without the presence of Byzantine failures), which can be used to obtain an estimate of the network size. The work of \cite{Chlebus_2009} focuses on the consensus problem under crash 
failures and assumes knowledge of $\log{n}$, where $n$ is the network size.


\section{Preliminaries}

\subsection{Computing Model and Problem Definition} \label{sec:model}

\textbf{The distributed computing model:} We consider a synchronous network represented by a graph $G$ whose nodes execute a distributed algorithm and whose edges represent connectivity in the network. The computation proceeds in synchronous rounds, i.e., we assume that nodes run at the same processing speed (and have access to a synchronized clock) and any message that is sent by some node $u$ to its neighbors in some round $r\ge 1$ will be received by the end of round $r$.\\

\textbf{Byzantine nodes:} Among the $n$ nodes ($n$ or its estimate is not known to the nodes initially), up to $B(n)$ can be \emph{Byzantine} and deviate arbitrarily from the given protocol. Throughout this paper, we assume that $B(n) = O(n^{1 - \delta})$ (where $n$ is the unknown network size), for any $\frac{3}{d} < \delta \leq 1$. We say that a node $u$ is \emph{honest} if $u$ is not a Byzantine node and use $\texttt{Honest}$ to denote the set of honest nodes in the network. Byzantine nodes are ``adaptive'', in the sense that they have complete knowledge of the entire states of all nodes at the beginning of every round (including random choices made by all the nodes),  and thus can take the current state of the computation into account when determining their next action (they also can know the future random choices of honest nodes). The Byzantine nodes  have unbounded computational power, and can deviate arbitrarily from the protocol. This setting is commonly referred to as the \emph{full information model}. We assume that the Byzantine nodes are {\em randomly} distributed in the network.\\

\textbf{Distinct IDs:} We assume that nodes (including Byzantine) have {\em distinct} IDs and they cannot lie about their ID while communicating with a neighbor. Note that the $n$ distinct IDs (where $n$ is the unknown network size) are assumed to be chosen from a {\em large space} (not known a priori to the nodes). Note that this precludes (most) nodes from estimating $\log n$ by potentially looking at the length of their IDs.\\

\textbf{Network Topology:} Let $G = (V, E)$ be the graph representing the network. We take $G$ to be the union of two other graphs $H$ and $L$ (both defined below). That is, $V(G) = V(H) = V(L) = V$, say, and $E(G)  =  E(H) \cup E(L)$. We take $H$ to be a sparse, random $d$-regular graph that is constructed by the union of $\frac{d}{2}$ (assume $d \geq 8$ is an even constant) random Hamiltonian cycles of $n$ nodes. We call this random graph model the \emph{$H(n, d)$ random graph model}. 
\footnote{We give more details on the model and analyze its properties in the appendix (Section \ref{section-locally-tree-like-property}).}
It is known that such a random graph is an expander with high probability. The $H(n, d)$ random, regular graph model is a well-studied and popular random graph model (see e.g., \cite{Wormald_1999}). In particular, the $H(n, d)$ random graph model has been used as a model for Peer-to-Peer networks and self-healing networks \cite{Law_2003, Pandurangan_2014}.

$E(L)$ is defined as follows. For $u, v \in V$, $(u, v) \in E(L)$ if and only if $dist(u, v) \leq k$ in $H$, where $k  =  \lceil\frac{d}{3}\rceil$ is a positive integer. In other words, each node has direct connections (via edges of $L$) to nodes that are within distance $k$. Note that adding the edges of $L$ makes $H$ a {\em ``small-world''} network, i.e., for each node $v$ in $G$, the neighbors of $v$ within distance $\frac{k}{2}$ in $H$ are connected to each other (thus the  clustering coefficient is increased in $G$ compared to $H$). The small-world property complements the expander property of the $d$-regular random graph, since the clustering coefficient of the random regular graph is  small. We exploit both properties crucially in the protocol. Larger the degree $d$, larger will be $k$, and large will be the robustness to Byzantine nodes, i.e., up to $O(n^{1 - \delta})$ Byzantine nodes can be tolerated where $\frac{3}{d} < \delta \leq 1$ (as defined earlier). Our small-world network is inspired by and  related to the Watts-Strogatz model \cite{Watts_1998, Barthelemy_1999}. However, it is important to note that the Watts-Strogatz model allows some nodes to have $\Theta(\log n)$ degree and hence not constant bounded degree, unlike our model. 

It is important to note that nodes in $G$ do not know a priori which edges are in $H$ and which are in $L$. However, as shown in Lemma \ref{obs-distinguish-between-H-and-L}, most (honest) nodes can distinguish between the two types of edges using a simple protocol.

We  point out although we assume a specific type of network model described above --- which, intuitively, is the worst case (most difficult) scenario for the algorithm designer due to its (essentially) identical local topological structure --- our results can be extended to apply to potentially any (sparse) graph that has high expansion and high clustering coefficient (e.g., one can presumably take any bounded-degree expander rather than a $d$-regular graph as $H$).\\

\textbf{Problem and Goal:} Our goal is to design a distributed protocol to estimate the number of nodes in $G$, even under the presence of a large number of Byzantine nodes. The problem is non-trivial, since each node has a local view and knowledge that is independent of the network size. We would like our protocol to run fast, i.e., in polylogarithmic (in the unknown network size $n$) rounds, and use only ``small-sized'' messages. A ``small-sized message'' is one  that contains a constant number of IDs and $O(\log{n})$ additional bits.

We now present the formal definition of the Byzantine counting problem. Since we assume a \emph{sparse} (constant bounded degree) network and a large number of Byzantine nodes, it is difficult for an algorithm where every honest node eventually knows the exact estimate of $n$. This motivates us to consider the following ``approximate, almost everywhere'' variant of counting:

\begin{definition}[Byzantine Counting] Suppose that there are $B(n)$ Byzantine nodes in the network. We say that an \emph{algorithm $A$ solves Byzantine Counting in $T$ rounds} if, in any run of $A$:
	\begin{enumerate}
		\item all honest nodes terminate in $T$ rounds,
		
		\item all except $B(n) + \epsilon n$ honest nodes (for any arbitrarily small constant $\epsilon > 0$) have a {\em constant factor} estimate of $\log n$ (i.e., if $\mathcal{L}$ is the 
		estimate, then $c_1 \log n \leq \mathcal{L} \leq c_2 \log n$, for some fixed positive constants $c_1$ and $c_2$), where $n$ is the actual network size.
	\end{enumerate}
\end{definition}
\subsection{Notations and a Few Necessary Definitions} \label{section-notation-and-preliminaries}

\begin{definition}
	For any two nodes $u$ and $v$ in $V$, the distance between them (in $G$) is defined as $dist_G(u, v) \defeq$ the length of a shortest path between $u$ and $v$ in $G$. Similarly, $dist_H(u, v)  
	\defeq$  the length of a shortest path between $u$ and $v$ in $H$.
\end{definition}
\begin{remark}
	For any node $v \in V(G)$, we follow the convention that $dist_G(v, v)  =  dist_H(v, v)  =  0$.
\end{remark}

\begin{definition}
	For any node $u$ and any set $V' \subset V(G)  =  V$, the distance between $u$ and $V'$ (in $G$) is defined as $dist_G(u, V') \defeq \text{min}\left\{dist_G(u, v)\ |\ v\in V' \right\}$.
\end{definition}
\begin{definition}
	For any two subsets $V'$ and $V''$ of $V(G)  =  V$, the distance between $V'$ and $V''$ (in $G$) is defined as $dist_G(V', V'') \defeq  \text{min}\left\{dist_G(u, v)\ |\ u\in V', v\in V''\right\}$.
\end{definition}

\begin{remark}
	In all our notations, the subscript $G$ or $H$ denotes the underlying graph. We will, however, for the most part, talk about $H$. Thus to obtain notational simplicity, we will omit the subscript 
	$H$ from now on. If at any point, we need to talk about $G$ instead, we will explicitly mention the subscript $G$. For example, $dist(u, v)$ will denote the length of a shortest path between $u$ 
	and $v$ in $H$, whereas $dist_G(u, v)$ will be used to denote the length of a shortest path between $u$ and $v$ in $G$. And so on.
\end{remark}


\begin{definition}\label{defn-metric-ball-preliminaries}
For any $v \in V(H)$ and any positive integer $r$, $B(v,r)$ is defined as the set of nodes within the ball of radius $r$ from $v$ (including at the boundary), i.e,
	\begin{center}
		$B(v, r)  \defeq  \left\{w \in V(H)\ |\ 0 \leq dist(v, w) \leq r\right\}$.
	\end{center}
\end{definition}

\begin{definition}\label{defn-metric-boundary-preliminaries}
	For any $v \in V(H)$ and any positive integer $r$, $Bd(v,r)$ is defined as the set of nodes at distance $r$ from $v$ (i.e., at the boundary),  i.e,
	\begin{center}
		$Bd(v, r)  \defeq  \left\{w \in V(H)\ |\ dist(v, w) = r\right\}$.
	\end{center}
\end{definition}

Next we introduce  the ``locally tree-like'' property of an $H(n,d)$ random graph: i.e., for most nodes $w$,  the subgraph  induced by $B(w,r)$ up to a certain radius $r$ looks ``like a tree''. This is 
stated more precisely as follows.

\begin{definition}\label{defn-typical-node-preliminaries}
	Let $G$ be an $H(n,d)$ random graph and $w$ be any node in $G$. Consider the subgraph induced by $B(w,r)$ for $r  =  \frac{\log{n}}{10\log{d}}$. Let $u$ be any node in $Bd(w,j)$, $1 \leq j < r$. 
	$u$ is said to be ``typical'' if $u$ has only one neighbor in $Bd(w,j-1)$ and $(d-1)$-neighbors in $Bd(w,j+1)$; otherwise it is called ``atypical''.
\end{definition}

\begin{definition}\label{defn-locally-tree-like-node-preliminaries}
	We call a node $w$ ``locally tree-like'' if no node in $B(w,r)$ is atypical. In other words, $w$ is ``locally tree-like'' if the subgraph induced by $B(w,r)$ is a $(d-1)$-ary tree.
\end{definition}

It can be shown using properties of the $H(n, d)$ random graph model and standard concentration bounds (cf.\ Section \ref{section-locally-tree-like-property}) that most nodes in $G$ are locally tree-like.
\begin{lemma}\label{lemma-most-nodes-are-locally-tree-like-preliminaries}
In an $H(n,d)$ random graph, with high probability, at least $n - O(n^{0.8})$ nodes are locally tree-like.
\end{lemma}

For the proof of this lemma, as well as for further details about the $H(n, d)$ random graph model, please refer to Section \ref{section-locally-tree-like-property}.

\begin{obs}\label{neighborhood-size-upper-bound-in-H}
	In a $d$-regular graph, for any vertex $v$, the number of vertices that are within a $\tau$-distance of $v$ is bounded by $(d-1)^{\tau+1}$, i.e., $|B(v, \tau)|  <  (d-1)^{\tau+1}$.
\end{obs}

Since any two vertices that are within $\tau$-distance of each other in $G$ is within $k\tau$ distance of each other in $H$ (which is a $d$-regular graph), we have that
\begin{obs}\label{neighborhood-size-upper-bound-in-G}
	In the graph $G$, for any vertex $v$, the number of vertices that are within a $\tau$-distance of $v$ is bounded by $(d-1)^{k\tau+1}$, i.e., $|B_G(v, \tau)|  <  (d-1)^{k\tau+1}$.
\end{obs}

\begin{definition}\label{definition-node-categorization}
We categorize the nodes in $V$ into the following distinct categories. Unlike our usual convention, the distances referred to in this definition refer to the respective distances in $G$ (not in $H$, as 
is usual).
\begin{enumerate}
	\item \textbf{Byzantine nodes:} The set of Byzantine nodes is denoted by $\texttt{Byz}$.
		
	\item \textbf{Honest nodes:} The set of honest nodes is defined to be $\texttt{Honest}   \defeq   V \setminus \texttt{Byz}$.
	
	\item \textbf{Locally tree-like nodes:} Please refer to Definition \ref{defn-locally-tree-like-node-preliminaries}. That is, the set of the locally tree-like nodes is defined as 
			$\texttt{LTL}  \defeq  \left\{v \in V\ |\ v\text{ is locally tree-like}\right\}$.
			
	\item \textbf{Non-locally-tree-like nodes:} The set of the non-locally-tree-like nodes is defined as $\texttt{NLT}  \defeq  V \setminus \texttt{LTL}$.
			
	\item \textbf{Unsafe nodes:} The set of nodes that have one or more $\texttt{NLT}$ nodes within a distance of $a\log{n}$, where $a  \defeq  \frac{\delta}{10k\log{(d-1)}}$. If we denote the set of 
			unsafe nodes by $\texttt{Unsafe}$, then
			\begin{center}
				$\texttt{Unsafe}   \defeq   \left\{v \in V\ |\ dist_G(v, \texttt{NLT})  \leq  a\log{n}\right\}$.
			\end{center}
	
	\item \textbf{Safe nodes:} Nodes that are not unsafe. In other words, the set of nodes that have no $\texttt{NLT}$ nodes within a distance of $a\log{n}$. If we denote the set of safe 
			nodes by $\texttt{Safe}$, then
			\begin{center}
				$\texttt{Safe}   \defeq   \left\{v \in V\ |\ dist_G(v, \texttt{NLT})  >  a\log{n}\right\}$.
			\end{center}
		
	\item \textbf{Bad nodes:} The set of bad nodes is defined to be $\texttt{Bad}   \defeq   \texttt{Byz}  \cup \texttt{NLT}$.
	

	\item \textbf{Byzantine-Unsafe nodes:} The set of nodes that have one or more bad nodes within a distance of $a\log{n}$, where $a  \defeq  \frac{\delta}{10k\log{(d-1)}}$. If we denote the set of 
			Byzantine-unsafe nodes by $\texttt{BUS}$, then
			\begin{center}
				$\texttt{BUS}   \defeq   \left\{v \in V\ |\ dist_G(v, \texttt{Bad})  \leq  a\log{n}\right\}$.
			\end{center}
		
	\item \textbf{Byzantine-Safe nodes:} Nodes that are not Byzantine-unsafe. In other words, the set of nodes that have no bad nodes within a distance of $a\log{n}$. If we denote the set of 
			Byzantine-safe nodes by $\texttt{Byz-safe}$, then
			\begin{center}
				$\texttt{Byz-safe}   \defeq   \left\{v \in V\ |\ dist_G(v, \texttt{Bad})  >  a\log{n}\right\}$.
			\end{center}
\end{enumerate}
\end{definition}

\begin{lemma}\label{lemma-set-sizes}
The various node sets defined in Definition \ref{definition-node-categorization} have the following sizes, respectively.
\begin{enumerate}
	\item $|\texttt{Byz}|  =  n^{1 - \delta}$.
	
	\item $|\texttt{Honest}|   =   n  -  n^{1 - \delta}$.
	
	\item $|\texttt{LTL}|   \geq   n  -  O(n^{0.8})$.
	
	\item $|\texttt{NLT}|   \leq   O(n^{0.8})$.
	
	\item $|\texttt{Unsafe}|  \leq  O(n^{0.8 + \frac{\delta}{10}})  =  o(n)$.
	
	\item $|\texttt{Safe}|  \geq  n - O(n^{0.8 + \frac{\delta}{10}})  =  n - o(n)$.
	
	\item $|\texttt{Bad}|   \leq   n^{1 - \delta}  +  n^{0.8}   \leq   2n^{1 - \delta}$ (assuming $\delta \leq 0.2$).
	
	
	\item $|\texttt{BUS}|   \leq   2(d-1)n^{1 - \frac{9\delta}{10}}   =   o(n)$.
	
	\item $|\texttt{Byz-safe}|   \geq   n  -  2(d-1)n^{1 - \frac{9\delta}{10}}   =   n - o(n)$.
\end{enumerate}
\end{lemma}
\begin{proof}

	\begin{enumerate}
		\item By definition.
		
		\item By definition.
		
		\item By Lemma \ref{lemma-most-nodes-are-locally-tree-like-preliminaries}.
		
		\item By Lemma \ref{lemma-most-nodes-are-locally-tree-like-preliminaries}.
		
		\item By Definition \ref{definition-node-categorization}, Observation \ref{neighborhood-size-upper-bound-in-G}, and Lemma \ref{lemma-most-nodes-are-locally-tree-like-preliminaries}.
		
		\item By Definition \ref{definition-node-categorization}, Observation \ref{neighborhood-size-upper-bound-in-G}, and Lemma \ref{lemma-most-nodes-are-locally-tree-like-preliminaries}.
		
		\item By definition.
		
		
		\item Follows from Observation \ref{neighborhood-size-upper-bound-in-G} and the definition of $\texttt{BUS}$.
		
		\item By definition.
	\end{enumerate}
	
\end{proof}

\begin{definition}\label{defn-parent-child}
	We call $u$ a child of $w$ with respect to $v$ (or $w$ the parent of $u$, with respect to $v$) if $u$ is a child of $w$ in the BFS tree rooted at $v$. Similarly, we call $u$ and $w$ siblings with 
	respect to $v$ if they are siblings in the BFS tree rooted at $v$. We note that, as is our usual custom, this BFS tree is in the graph $H$ and not in $G$.
\end{definition}

It is important to note that nodes in $G$ do not know a priori which edges are in $H$ and which are in $L$. However, the following lemma assures us that most (honest) nodes can distinguish between the 
two types of edges using a simple protocol.

\begin{lemma}\label{obs-distinguish-between-H-and-L}
	For any honest node $v$, if $v$ has no Byzantine neighbor in $G$ (that is, no Byzantine neighbor in its $k$-distance neighborhood in $H$), then $v$ can faithfully reconstruct the topology of its 
	$k$-distance neighborhood in $H$ from the information it is provided by its $G$-neighbors.
\end{lemma}
\begin{proof}
	For any $x \in V(G)$, let $N_G(x)$ denote the set of $G$-neighbors of $x$.
	Let $w$ and $u$ be two $G$-neighbors of $v$. Then we observe that
	\begin{itemize}
		\item $w$ is a child of $u$ with respect to $v$ if and only if $N_G(w) \cap N_G(v)  \subset  N_G(u) \cap N_G(v)$.
		\item $u$ is a child of $w$ with respect to $v$ if and only if $N_G(u) \cap N_G(v)  \subset  N_G(w) \cap N_G(v)$.
		\item $u$ and $v$ are siblings if $u \in N_G(w)$ and $w \in N_G(u)$ but neither of them is a child of the other.
	\end{itemize}
\end{proof}

\begin{remark}
	Since $d$ and $k$ are constants, the list of neighbors is still $O(1)$, and hence can be exchanged in a constant number of rounds (using small sized messages).
\end{remark}


\section{The Algorithm and its analysis}  \label{section-algorithm}

For the sake of exposition, we first describe the algorithm and analyze its behavior \emph{free from} the influence of any Byzantine nodes; in other words, we will assume that all nodes (including 
Byzantine nodes) honestly execute the protocol without malicious behavior. We will discuss any malicious effects the Byzantine nodes may have in Section \ref{section-algorithm-with-byz} and describe how 
to modify the algorithm (and analysis)  to counter the Bzyantine nodes.

\subsection{Description of the algorithm (assuming no influence of Byzantine nodes)}
\textbf{Phases and subphases:} This is a distributed algorithm that runs in \emph{phases}. In the $i^{\text{th}}$ phase, the algorithm works with the current estimation of $\log{n}$, which is $i$. We 
reserve the letter $i$ exclusively to denote the phase that the algorithm is presently in. For $i \geq 1$, the $i^{\text{th}}$ phase consists of several runs (repetitions) of the same random experiment 
(the random experiment is described in the next few paragraphs; also see Lines \ref{color1} through \ref{subphase-ends} of the pseudocode in Algorithm \ref{alg}). We call one such run a \emph{subphase} 
of the $i^{\text{th}}$ phase. We would usually index the subphases by $j$, i.e., we will very frequently use the phrase ``in the $j^{\text{th}}$ subphase of the $i^{\text{th}}$ phase'' in our 
description and analysis of the algorithm. We note that in a synchronized network the value of $i$ and $j$ are known to all nodes.

The $i^{\text{th}}$ phase consists of exactly $\alpha_i$ subphases (repetitions), where $\alpha_i  \defeq  \lceil\frac{\log{(\frac{1}{\epsilon})} + i + 1 - \log{d}}{(i-2)\log{(d-1)}}\rceil$. We call 
$\epsilon$ \emph{the error parameter} --- by changing its value (please refer to Line \ref{value-of-alpha-i} in the pseudocode in Algorithm \ref{alg}), we (the algorithm designer) can control exactly 
how large a fraction of the honest nodes would estimate $\log{n}$ correctly (i.e., get a constant-factor approximation of $\log{n}$). Theorem \ref{theorem-main-result-of-the-paper}, which is the main 
result of this paper, tells us that at most $\epsilon$-fraction of the honest nodes would \emph{fail} to get a constant factor approximation of $\log{n}$.\\

\textbf{Basic idea} (see also Section \ref{sec:technical})\textbf{:} In one random experiment, i.e., in the $j^{\text{th}}$ subphase of the $i^{\text{th}}$ phase, say, every node sends out some tokens 
(these contain some information) that propagate through the network (by flooding) for some (pre-determined) number of steps (rounds), at the end of which every node takes stock of the tokens it has 
received over the intermediate rounds.\\

\textbf{Color of a token:} Every token circulating in the network has a \emph{color} (defined next), which is passed down to a token from its generating node. Every node $v$ tosses an unbiased coin until 
it gets its first head (see Line \ref{color1} of the pseudocode in Algorithm \ref{alg}). If a node $v$ gets its first head at the $r^{\text{th}}$ trial, we call $r$ to be the \emph{color} of $v$ (see 
Line \ref{color2} of the pseudocode in Algorithm \ref{alg}). Thus the \emph{color} of a node is always a positive integer, which may be (but is not necessarily) different for different nodes.\\

\textbf{Estimating $\log{n}$:} When $i$ is much smaller than $\log{n}$, most nodes will receive their respective highest colored tokens in the last round. In contrast, when $i$ is of the same order as 
$\log{n}$, most nodes will have received their respective highest colored tokens much before the last round. This provides a node with a way to determine when its estimate of $\log{n}$, which is $i$, 
has reached close to the actual value of $\log{n}$.
\begin{algorithm}[h]
\begin{algorithmic}[1]
	\State Ask all the neighbors for their respective adjacency lists and distinguish between the edges of $H$ and $L$ from that information. 
		
	\For{$i \gets 1,2,\dots$} \label{begin-with-1} \MyComment{$i$ denotes the phase node $v$ is in}
		\State $FlagTerminate \gets 1$
		
		\If{$d(d - 1)^{i - 2}   \leq   \frac{2}{\epsilon}$}
            \State $\alpha_i  \gets  \lceil\frac{\log{(\frac{1}{\epsilon})} + i + 1}{\log{d} + (i-2)\log{(d-1)} - 1}\rceil$ \label{value-of-alpha-i-case-1} \MyComment{$0 < \epsilon < 1$ is the error-parameter}
        \Else
            \State $\alpha_i   \gets   1  +  \frac{i + 1}{\log{(\frac{1}{\epsilon})}}$ \label{value-of-alpha-i-case-2}
        \EndIf
		
		\For{$j \gets 1,2,\dots,i\alpha_i$} \label{phase-i-has-alpha-i-subphases} \MyComment{Phase $i$ consists of $\alpha_i$ subphases; the subphases are indexed by $j$}
			\State $\text{$v$ tosses a fair coin until the outcome is heads in the $r$-th trial, for some $r \geq 1$.}$ \label{color1}
			\State $c_{v, i} \gets r$ \label{color2}
			\State Flood the color $c_{v, i}$, along the edges of $H$ only, for exactly $i$ steps. \MyComment{This is possible by virtue of Lemma \ref{obs-distinguish-between-H-and-L}.}
			\For{time $t = 1, 2, \ldots, i$}
				\State In each round $t$, mark and store the highest color received. Let's call it $k_t$\label{received-highest-color}
			\EndFor
			
			\If{$k_i > k_t$, $\forall 1 \leq t < i$, and $k_i   >   \log{(d(d-1)^{i-1})}  -  \log{\log{(d(d-1)^{i-1})}}$}\label{criterion-for-continuing-first-line}
				\State $FlagTerminate \gets 0$ \label{criterion-for-continuing}
			\EndIf \label{subphase-ends}
		\EndFor
		
		\If{$FlagTerminate = 1$}
			\State Decide $i$ and terminate all for-loops. \MyComment{$v$ accepts $i$ as the estimate of $\log{n}$}\label{final-accept}
		\Else
			\State Continue to the next phase $i+1$.
		\EndIf
	\EndFor
\end{algorithmic}
\caption{The basic counting algorithm (in the absence of Byzantine nodes). Code for node $v$.}
\label{alg}
\end{algorithm}


\subsection{Analysis of the algorithm (assuming Byzantine nodes behave honestly)} \label{section-analysis}

In this section we show that the algorithm gives a $(\frac{b}{a})$-factor approximation of $\log{n}$ with high probability, where $a  \defeq  \frac{\delta}{10k\log{(d-1)}}$ and $b   \defeq   
\frac{4}{\log{(1 + \frac{h}{d})}}$, where $h$ is the edge-expansion of $H$. Note that $0 < a < b < 1$. We recall that $n^{1-\delta}$ is the number of Byzantine nodes in the network $G$, and $d$ is the uniform degree of $H$. ($H$ is a subset of $G$. For the exact definition of $H$, please refer to Section \ref{sec:model}.)

\begin{obs}\label{b-is-large}
	$b\log{n} \geq 2D(H)$, where $D(H)$ is the diameter of $H$.
\end{obs}

\paragraph*{High-level overview of the proof} We break our analysis up into two different stages of the algorithm. We show that the following statements hold with high probability.
\begin{enumerate}
	\item For $i  <  a\log{n}$, at least $(1 - \epsilon)$-fraction of the good nodes do not accept $i$ to be the right estimate of $\log{n}$, and they continue with the algorithm. The rest of the nodes --- i.e., at most $\epsilon$-fraction of the good nodes --- even though they have stopped generating tokens, still continue to forward tokens generated by other nodes. $0 < \epsilon \leq 1$ is a constant and we can make it arbitrarily small.
	
	\item If $i = b\log{n}$, all but $o(n)$ of the remaining active nodes accept $i$ to be the estimate of $\log{n}$ and they stop producing tokens. They however continue to forward tokens generated by other (if any) nodes.
\end{enumerate}

We cannot say which way a node will decide when $a\log{n} \leq i < b\log{n}$. The above two statements are, however, sufficient to give us an approximation factor of $\frac{b}{a}   =   
\frac{40k\log{(d-1)}}{\delta\log{(1 + \frac{h}{d})}}$.

\subsubsection{When $i$ is small: In particular, when $i < a\log{n}$} \label{subsection-i-is-small}

For the sake of the analysis in this subsection only, we will consider only \emph{safe nodes}, i.e., only those nodes in the set $\texttt{Safe}$. Let us first take note of a few properties of the geometric distribution though; these will be useful later.

\begin{obs}\label{individual-color-probabilities}
	For any node $v$ and any positive integer $r$,
	\begin{enumerate}
		\item $Pr[c_v = r] = \frac{1}{2^r}$.
		\item $Pr[c_v \geq r] = \frac{1}{2^{r-1}}$.
		\item $Pr[c_v < r]  =  1 - Pr[c_v \geq r]  =  1 - \frac{1}{2^{r-1}}$.
		\item $Pr[c_v \leq r] = 1 - Pr[c_v \geq r+1]  =  1 - \frac{1}{2^r}$.
		\item $Pr[c_v > r]  =  1 - Pr[c_v \leq r]  =  \frac{1}{2^r}$.
	\end{enumerate}
\end{obs}

For any non-empty $V' \subset V(G)$, $c^{\text{max}}_{V'}$ is defined as $c^{\text{max}}_{V'}  \defeq  \left\{c_v\ |\ v \in V'\right\}$. Suppose $|V'| = n'$.
\begin{obs}\label{max-color-probabilities}
	For any positive integer $j$,
	\begin{enumerate}
		\item $Pr[c^{\text{max}}_{V'}  <  r]   =   (Pr[c_v < r])^{n'}   =   (1 - \frac{1}{2^{r-1}})^{n'}$.
		\item $Pr[c^{\text{max}}_{V'}  \geq  r]   =   1 - Pr[c^{\text{max}}_{V'}  <  r]   =   1  -  (1 - \frac{1}{2^{r-1}})^{n'}$.
		\item $Pr[c^{\text{max}}_{V'}  \leq  r]   =   Pr[c^{\text{max}}_{V'}  <  r+1]   =   (1 - \frac{1}{2^r})^{n'}$.
		\item $Pr[c^{\text{max}}_{V'}  >  r]   =   Pr[c^{\text{max}}_{V'}  \geq  r+1]   =   1  -  (1 - \frac{1}{2^r})^{n'}$.
		\item $Pr[c^{\text{max}}_{V'}  =  r]   =   Pr[c^{\text{max}}_{V'}  \geq  r] - Pr[c^{\text{max}}_{V'}  >  r]   =   (1 - \frac{1}{2^{r}})^{n'} - (1 - \frac{1}{2^{r-1}})^{n'}$.
	\end{enumerate}
\end{obs}

\begin{lemma}\label{max-color-upper-bound}
	$Pr[c^{\text{max}}_{V'}  >  2\log{n'}]  \leq  \frac{1}{n'}$.
\end{lemma}
\begin{proof}
	\begin{align*}
		&Pr[c^{\text{max}}_{V'}  >  2\log{n'}]   =   1  -  (1 - \frac{1}{2^{2\log{n'}}})^{n'}   =   1  -  (1 - \frac{1}{{n'}^2})^{n'}\\
		&\leq   1  -  (1 - \frac{n'}{{n'}^2})   =   1   -   (1 - \frac{1}{n'})   =   \frac{1}{n'}\text{.}
	\end{align*}
\end{proof}

\begin{lemma}\label{max-color-lower-bound}
	$Pr[c^{\text{max}}_{V'}  \leq  \log{n'} - \log{\log{n'}}]   <   \frac{1}{n'}$.
\end{lemma}
\begin{proof}
	\begin{align*}
		&Pr[c^{\text{max}}_{V'}  \leq  \log{n'} - \log{\log{n'}}]   =   (1 - \frac{1}{2^{\log{n'} - \log{\log{n'}}}})^{n'}\\
		&=   (1 - \frac{\log{n'}}{n'})^{n'}   \leq   \text{exp}(-\frac{(\log{n'}) . n'}{n'})   =   \text{exp}(-\log{n'})   <   \frac{1}{n'}\text{.}
	\end{align*}
\end{proof}

For any node $v$ and any non-negative integer $r$, let us denote the set $B(v, r) \setminus \left\{v\right\}$ by $B^*(v, r)$. We recall that from the ``locally tree-like property'' (cf. Definition 
\ref{defn-locally-tree-like-node-preliminaries} and Lemma \ref{lemma-most-nodes-are-locally-tree-like-preliminaries}), for any \emph{safe node} $v$,
\begin{align*}
	&|B(v, r)|  =  1 + d\cdot\textstyle{\sum_{j=1}^r} (d-1)^{j-1}  \implies  |B^*(v, r)|  =  d\sum_{j=1}^r (d-1)^{j-1}   =   \frac{d(d-1)^r}{d-2}\text{, and}\\
	&|Bd(v, r)|  =  d(d-1)^{r-1}\text{.}
\end{align*}

For any positive integer $r$, let $l_r  \defeq  \log{d} + r\log{(d-1)}$. We observe that $l_r = l_{r-1} + \log{(d-1)}$. Then
\begin{lemma}\label{ball-boundary-expected-max-color-values}
	$\log{(|B^*(v, r)|)}  =  l_r - \log{(d-2)}$ and $\log{(|Bd(v, r)|)}  =  l_r - \log{(d-1)}$.
\end{lemma}

\begin{lemma}\label{lemma-inner-ball-probability}
	$Pr[c^{\text{max}}_{B^*(v, r)}  >  2(l_r - \log{(d-2)})]   \leq   \frac{d-2}{d(d-1)^r}$.
\end{lemma}

\begin{proof}
	Follows from Lemma \ref{max-color-upper-bound} and Lemma \ref{ball-boundary-expected-max-color-values}.
\end{proof}

\begin{lemma}\label{lemma-ball-boundary-probability}
	$Pr[c^{\text{max}}_{Bd(v, r)}   \leq   l_r - \log{(d-1)} - \log{(l_r - \log{(d-1)})}]   <   \frac{1}{d(d-1)^{r-1}}$.
\end{lemma}
\begin{proof}
	Follows from Lemma \ref{max-color-lower-bound} and Lemma \ref{ball-boundary-expected-max-color-values}.
\end{proof}

Next we show that the probability that a \emph{safe node} decides to stop (when $i < a\log{n}$) is bounded by a constant (any arbitrarily small, but fixed constant).

\begin{lemma} \label{lemma-probability-of-error-in-the-lower-end}
	$Pr[\text{a {\emph{safe}} node }v\text{ makes a wrong decision in the }i^{\text{th}}\text{phase}]   <   \frac{\epsilon}{2^{i+1}}$.
\end{lemma}

We will use a series of other, smaller results to show the above. One subtle issue to keep in mind is that since the ``failure probability'' for a safe node is \emph{not} $0$ (zero), there may be some safe nodes that do decide wrongly. Those safe nodes, in turn, will no longer generate tokens in the following phases. Therefore, when we calculate the failure probability for a safe node $v$, say, in phase $i$, we have to take into consideration the fact that there may be some nodes in the $i$-hop neighborhood of $v$, i.e., in $B(v, i)$, that made a wrong decision in some previous phase $j$, $j < i$, and is thus \emph{inactive} in phase $i$.

We show this by induction on $i$, where $i$ is the phase-number. We note that in the very first phase, all the nodes are active, thus there is no need to consider inactive nodes. This helps us prove the basis of the induction. Next we assume that for any $i < \log{n}$, the probability that a safe node $v$ went inactive in some previous phase $i'$, $i' < i$, is at most $\frac{\epsilon}{2^{i' + 1}}$, where $\epsilon$ is the error parameter. This forms the inductive hypothesis. Assuming this, we go on to show that the failure probability for a safe node in the $i^{\text{th}}$ phase is less than $\frac{\epsilon}{2^{i + 1}}$.

We defer the detailed, formal proof to the appendix --- please refer to Section \ref{section-proof-of-induction-analysis-in-the-lower-end}.
\paragraph{Translating the constant probability of error into a ``low'' probability of error.}

Lemma \ref{lemma-probability-of-error-in-the-lower-end} promises us that any individual node has a small probability of error when $i < a\log{n}$. So the expected number of nodes to make an error is also small. We, however, want to show a high probability bound on the number of nodes that make a mistake.

In order to show that, we proceed along the usual way of formulating an indicator random variable and then computing the expectation of the sum of the individual indicator random variables by using the principle of linearity of expectation. We show the high probability bound by using the method of \emph{bounded differences} (Azuma's Inequality, more specifically).

Now to the formal description.\\

Let $Y^v_i$ be an indicator random variable which is $1$ if and only if $v$ decides $i$ to be a correct estimate of $\log{n}$. Lemma \ref{lemma-probability-of-error-in-the-lower-end} shows that 
$Pr[Y^v_i = 1]  <  \frac{\epsilon}{2^{i+1}}$. Now let $$Y_i  =  \sum_{v\in V}{Y^v_i}\text{.}$$ That is, $Y_i$ denotes the number of nodes that decide \emph{wrongly} in the $i^{\text{th}}$ phase. We recall once again that here we are interested only in the case where $i < a\log{n}$. Then $Y_i$ cannot be too large, i.e., not too many nodes can decide wrongly in one phase.
\begin{lemma}\label{not-too-many-nodes-decide-wrongly-in-the-lower-end}
	$Pr[Y_i > \frac{n\epsilon}{2^i}] < \frac{1}{n^4}$ if $i < a\log{n}$.
\end{lemma}
\begin{proof}
\begin{align*}
	&E[Y_i]  =  E[\sum_{v\in V}{Y^v_i}]  =  \sum_{v\in V}E[{Y^v_i}]\text{  [by linearity of expectation]}\\
	&=  \sum_{v\in V}{Pr[Y^v_i = 1]}\text{  [since }Y^v_i\text{ is an indicator random variable]}\\
	&<  \sum_{v\in V}{\frac{\epsilon}{2^{i+1}}}  =  \frac{n\epsilon}{2^{i+1}}
\end{align*}

Two vertices $v$ and $w$ are independent if their $i$-distance neighborhoods do not intersect, i.e., if the distance between them is greater than $2i$. In other words, $v$ going defective can affect only those vertices that are within a distance of $2i$ to $v$. The number of vertices that are within a $2i$ distance of $v$ is at most $(d - 1)^{(2i + 1)}$. By the Azuma-Hoeffding Inequality \cite{Dubhashi_2009},
\begin{align*}
	&Pr[Y_i - E[Y_i] \geq \frac{n\epsilon}{2^{i+1}}]  \leq  \text{exp}(-\frac{{(\frac{n\epsilon}{2^{i+1}})}^2}{2n\cdot (d - 1)^{2(2i + 1)}})   =   
	\text{exp}(-\frac{n\epsilon^2}{2^{2i+3} . (d-1)^{4i+2}})\\
	&=   \text{exp}(-\frac{n\epsilon^2}{2^k})\text{, say, where }k = 2i+3 + (4i+2)\log{(d-1)}\text{.}
\end{align*}
Since, $i   <  a\log{n}   =   \frac{\delta\log{n}}{10\log{(d-1)}}$,
\begin{align*}
	&k   <   \frac{\delta\log{n}}{5\log{(d-1)}}  +  3 + 2\log{(d-1)}  +  \frac{2\delta\log{n}}{5}\\
	&=   \frac{\delta\log{n}}{5\log{(d-1)}} (2\log{(d-1)} + 1)  +  2\log{(d-1)} + 3\\
	&<   \log{n} - \log{\log{n}} - 2 - 2\log{(\frac{1}{\epsilon})}\\
	&\text{[assuming}\log{n}  >  \frac{5\log{(d-1)} (\log{\log{n}} + 2\log{(d-1)} + 3 + 2\log{(\frac{1}{\epsilon})} + 2)}{(5 - 2\delta)\log{(d-1)} - \delta}\text{,}\\
	&\text{which is true for large enough values of }n\text{]}
\end{align*}
Thus
\begin{align*}
	&Pr[Y_i - E[Y_i]  \geq  \frac{n\epsilon}{2^{i+1}}]   \leq   \text{exp}(-\frac{n\epsilon^2}{2^k})   \leq   
	\text{exp}(-\frac{n\epsilon^2}{2^{\log{n} - \log{\log{n}} - 2 - 2\log{(\frac{1}{\epsilon})}}})\\
	&=   \text{exp}(-\frac{n\epsilon^2}{\frac{n}{4\log{n}\cdot\frac{1}{\epsilon^2}}})   =   \text{exp}(-4\log{n})   <   \frac{1}{n^4}
\end{align*}

But again $E[Y_i] < \frac{n\epsilon}{2^{i+1}}$. Hence $Pr[Y_i > \frac{n\epsilon}{2^i}] \leq Pr[Y_i - E[Y_i] \geq \frac{n\epsilon}{2^{i+1}}] < \frac{1}{n^4}$.
\end{proof}

Now this is true for one particular phase $i$. Summing over all the phases (recall that we are concerned here only with the case $i < a\log{n}$), we get that the fraction of nodes that make a wrong decision cannot be more than $$\sum_{i < a\log{n}}\frac{\epsilon}{2^i} < \epsilon\text{,}$$ and this is true with probability $$>   (1  -  \sum_{i < a\log{n}}\frac{1}{n^4})   >   (1  -  \frac{1}{n^3})\text{.}$$

Thus we have

\begin{lemma}\label{theorem-lower-end}
	For Algorithm 1, the following holds with probability $>  1 - \frac{1}{n^3}$: While $1 \leq i < a\log{n}$, at most $\epsilon$-fraction of the nodes decide wrongly, i.e., decide $i$ to be a 
	correct estimate of $\log{n}$ (where $\epsilon$ is any arbitrarily small but fixed positive constant).
\end{lemma}
\begin{proof}
	Follows from Lemma \ref{not-too-many-nodes-decide-wrongly-in-the-lower-end} and Lemma \ref{lemma-set-sizes}.
\end{proof}
\subsubsection{When $i = \Theta(\log{n})$: In particular, when $i  =  b\log{n}$}\label{subsection-i-is-large}

Here we show that the following statement holds with probability at least $1 - \frac{1}{n^2}$: If a node $v$ is still active at the beginning of this phase, by the end of this phase, it accepts the 
current value of $i$, i.e., $b\log{n}$, to be a correct estimate of $\log{n}$ and terminates.

\begin{lemma}\label{lemma-upper-end-no-high-color}
	The following holds with probability at least $1 - \frac{1}{n^2}$: In all the $i\alpha_i$ subphases of phase $i$ (where $i = b\log{n}$), it is always the case that 
	$c^{\text{max}}_{V} \leq 4\log{n} - 1$, where $c^{\text{max}}_{V}  \defeq  \left\{c_v\ |\ v \in V\right\}$, i.e., the highest color generated in the network.
\end{lemma}
\begin{proof}
	From Observation \ref{individual-color-probabilities}, for any particular node $w$, $Pr[c_w > 4\log{n} - 1]  =  \frac{1}{2^{4\log{n} - 1}}  =  \frac{2}{n^4}$. Taking the union bound over all 
	$w \in V(G)$,
	\begin{equation}
		Pr[c^{\text{max}}_{V} > 4\log{n} - 1]  \leq  \frac{2}{n^3}
	\end{equation}

	This is for one subphase of the $i^{\text{th}}$ phase. Since there are $i\alpha_i$ subphases in the $i^{\text{th}}$ phase, we take the union bound over all the subphases and get that with 
	probability at least $1 - \frac{2i\alpha_i}{n^3}$, $c^{\text{max}}_{V} \leq 4\log{n} - 1$ holds in \emph{all} the $i\alpha_i$ subphases. But $\frac{2i\alpha_i}{n^3}  =  
	\frac{\Theta(\log^2{n})}{n^3}  <  \frac{1}{n^2}$. Thus with probability at least $1 - \frac{1}{n^2}$, $c^{\text{max}}_{V} \leq 4\log{n} - 1$ holds in \emph{all} the $i\alpha_i$ subphases.
\end{proof}

\begin{lemma}\label{theorem-upper-end}
	The following holds with probability at least $1 - \frac{1}{n^2}$ for Algorithm 1: If a node $v$ is still active at the beginning of phase $i$ (when $i = b\log{n}$), by the end of this phase, it 
	accepts the current value of $i$, i.e., $b\log{n}$, to be a correct estimate of $\log{n}$ and terminates.
\end{lemma}
\begin{proof}
	We recall that in order for an honest node $v$ to continue after this phase, the following criterion must be satisfied at least once in the $i\alpha_i$ subphases of the $i^{\text{th}}$ phase 
	(Please see Line \ref{criterion-for-continuing} of the pseudocode):
	\begin{center}
		$k_i > \log{d}+(i-1)\log{(d-1)}-\log{(\log{d}+(i-1)\log{(d-1)})}$,
	\end{center}
	where $k_i$ is the highest color that $v$ receives after $i$ rounds, i.e., at the end of the $j^{\text{th}}$ subphase of the $i^{\text{th}}$ phase. Substituting $i  =  b\log{n}  =  
	\frac{4\log{n}}{\log{(1 + \frac{h}{d})}}  >  4\log{n}$, we get that in order for an honest node $v$ to continue after this phase, the following criterion must be satisfied at least once in the 
	$i\alpha_i$ subphases of the $i^{\text{th}}$ phase:
	\begin{align*}
		&k_i   >   \log{d}+(i-1)\log{(d-1)}-\log{(\log{d}+(i-1)\log{(d-1)})}\\
		&>   \log{d}+(4\log{n}-1)\log{(d-1)}-\log{(\log{d}+(4\log{n}-1)\log{(d-1)})}\\
		&>   \frac{1}{2} . (\log{d}+(4\log{n}-1)\log{(d-1)})   >   (4\log{n} - 1) . \frac{\log{(d-1)}}{2}\\
		&\geq   4\log{n} - 1\text{, assuming }\log{(d-1)}  \geq  2\text{, or equivalently, }d \geq 5\text{.}
	\end{align*}

	By Lemma \ref{lemma-upper-end-no-high-color}, with probability at least $1 - \frac{1}{n^2}$, no node generates a color $> 4\log{n} - 1$. Therefore $v$ will not receive any such color 
	($> 4\log{n} - 1$) either. So in all the $i\alpha_i$ subphases of phase $i$ (where $i = b\log{n}$), $k_i$ will always be $\leq  4\log{n} - 1$ and therefore $v$ will \emph{not} continue after this 
	phase.
\end{proof}


\subsection{The Byzantine Protocol: Modifications to Algorithm \ref{alg} and its analysis} \label{section-algorithm-with-byz}

We next discuss the modifications  made to the Basic Counting Protocol (Algorithm \ref{alg}) to counter the effect of the Byzantine nodes --- this gives us the Byzantine Counting Protocol (Algorithm \ref{byz-alg}).

\subsubsection{Description of the modifications in the algorithm}
\begin{enumerate}
	\item At the very beginning (that is, even before phase $1$ starts), every honest node $v$ asks its neighbors in $G$ for their own IDs and the IDs of their respective neighbors. We observe that this takes a constant number of rounds. From that neighborhood information of its neighbors, $v$ tries to reconstruct the topology of its $k$-distance neighborhood in $H$. Lemma 
	\ref{obs-distinguish-between-H-and-L} tells us that this is possible when there are no Byzantine nodes.
	
	When there are Byzantine nodes, however, they can try to provide false neighborhood data to $v$. The algorithm dictates that $v$ shuts itself down (that is, goes into \emph{crash failure}) if $v$ receives inconsistent or conflicting data from two or more of its neighbors. Please refer to Line \ref{shut-yourself-down} of the pseudocode in Algorithm \ref{byz-alg}.

	\item For every color that $v$ receives from a neighbor $w$, say, $v$ checks (via the \emph{lattice edges}, i.e., the edges of $L$) with all the nodes in $B(w, k-1)$ (this ball $B$ is 
	defined with respect to $H$) to verify that $w$ indeed received that color via a legitimate path (up to a distance of $k-1$) from its $(k-1)$-distance neighbors in $H$. Please refer to Line 
	\ref{verify-along-edges-of-L} of the pseudocode in Algorithm \ref{byz-alg}.
	
	We note a minor detail here: For colors received within the first $t$ time-steps (in any $j^{\text{th}}$ subphase of any phase $i$), when $1 \leq t \leq k-1$, an honest node $v$ checks with nodes in the smaller ball $B(w, t)$ (instead of $B(w, k-1)$).
	
	Lemma \ref{lemma-no-chain-of-length-k-even-in-pretend} guarantees that (for any honest node $v$) the Byzantine nodes cannot fool $v$ into believing the existence of a $k$-length chain, composed purely of Byzantine nodes, in its $k$-distance neighborhood in $H$. Thus it ensures that a Byzantine node is \emph{not} able to push any arbitrary color into the network without raising a flag.
\end{enumerate}
\newpage
\subsubsection{The Pseudocode for the Byzantine counting algorithm}

Lines \ref{shut-yourself-down} and \ref{verify-along-edges-of-L} respectively indicate the changes from the previous algorithm (please refer to Algorithm~ \ref{alg}). These lines are shown in boldface. Suppose that a node sends a message with some color $c$. We say that $c$ is a \emph{legitimate color} if it was generated by an honest node. Note that some nodes might be forwarding colors generated by Byzantine nodes.

\begin{algorithm}[h]
\begin{algorithmic}[1]
	\State Ask all the neighbors (in $G$) for their respective adjacency lists and distinguish between the edges of $H$ and $L$ from that information.
	\State \textbf{If $v$ gets conflicting or contradictory information from two or more of its neighbors in $G$, $v$ shuts down, i.e., $v$ goes into \emph{crash failure}.} \label{shut-yourself-down}

	\For{$i \gets 1,2,\dots$} \label{begin-with-1-with-byz} \MyComment{$i$ denotes the phase node $v$ is in}
		\State $FlagTerminate \gets 1$
		
		\If{$d(d - 1)^{i - 2}   \leq   \frac{2}{\epsilon}$}
            \State $\alpha_i  \gets  \lceil\frac{\log{(\frac{1}{\epsilon})} + i + 1}{\log{d} + (i-2)\log{(d-1)} - 1}\rceil$ \MyComment{$0 < \epsilon < 1$ is the error-parameter}
        \Else
            \State $\alpha_i   \gets   1  +  \frac{i + 1}{\log{(\frac{1}{\epsilon})}}$
        \EndIf
		
		\For{$j \gets 1,2,\dots,i\alpha_i$}
			\State $\text{$v$ tosses a fair coin until the outcome is heads in the $r$-th trial, for some $r \geq 1$.}$ \label{color1-with-byz}
			\State $c_{v, i} \gets r$ \label{color2-with-byz}
			\State Flood the color $c_{v, i}$, along the edges of $H$ only, for exactly $i$ steps.
			\For{time $t = 1, 2, \ldots, i$}
				\State \textbf{In each round $t$, for every received color $c$, if $v$ got $c$ from its neighbor (in $H$) $w$, say, $v$ checks with the $(k - 1)$-distance neighbors (in $H$) of 
					$w$ to verify that $c$ is a legitimate color.} \label{verify-along-edges-of-L}
				\State In each round $t$, mark and store the highest color received. Let's call it $k_t$\label{received-highest-color-with-byz}
			\EndFor
			
			\If{$k_i > k_t$, $\forall 1\leq t<i$, and $k_i   >   \log{d}+(i-1)\log{(d-1)}  -  \log{(\log{d}+(i-1)\log{(d-1)})}$}
				\State $FlagTerminate \gets 0$ \label{criterion-for-continuing-with-byz}
			\EndIf
		\EndFor
		
		\If{$FlagTerminate = 1$}
			\State Decide $i$ and terminate all for-loops. \MyComment{$v$ accepts $i$ as the estimate of $\log{n}$}\label{final-accept-with-byz}
		\Else
			\State Continue to the next phase $i+1$.
		\EndIf
	\EndFor
\end{algorithmic}
\caption{The Byzantine counting algorithm. Code for an honest node $v$.}
\label{byz-alg}
\end{algorithm}


\subsection{Analysis of the algorithm}\label{section-analysis-with-byz}

\subsubsection{Preliminaries --- Some useful observations}

Let $\texttt{Crashed}$ be the set of honest nodes that shut themselves down at the very beginning of the algorithm (please see Line \ref{shut-yourself-down} of the pseudocode in Algorithm \ref{byz-alg}). 
Let $\texttt{Core}$ be the largest connected component in $H$ induced by $\texttt{Uncrashed}$, where $\texttt{Uncrashed}   \defeq   \texttt{Honest}  \setminus  \texttt{Crashed}$.

\begin{lemma}[See \cite{Augustine_2015_enabling_robust_and_efficient}]\label{lemma-core-is-an-expander}
	$\texttt{Core}$ has size at least $n - o(n)$. Moreover, $\texttt{Core}$ is an expander with edge-expansion at least $\gamma$, where $\gamma > 0$ is a constant.
\end{lemma}
\begin{proof}
	Follows from Lemma $3$ in \cite{Augustine_2015_enabling_robust_and_efficient}.
\end{proof}

\begin{obs}\label{no-Byzantine-chain-of-length-k}
	In the graph $H$, with high probability, there is no chain of length $\geq k$ composed of Byzantine nodes only.
\end{obs}
\begin{proof}
	We have that $k  =  \lceil\frac{d}{3}\rceil$ and $\delta > \frac{3}{d}$, implying $k\delta > 1$. We assume that $k\delta = 1 + \delta'$ for a fixed positive constant $\delta'$.
	
	The number of possible $k$-length chains is upper-bounded by $n \cdot d^{k-1}$. We recall that the Byzantine nodes are randomly distributed in the network. Therefore, for any one such chain, the 
	probability that it is composed purely of Byzantine nodes is $(\frac{n^{1-\delta}}{n})^k  =  n^{-k\delta}$. By union bound, the probability that there is at least one chain made only of Byantine 
	nodes is upper-bounded by
	\begin{align*}
		&n \cdot d^{k-1} \cdot n^{-k\delta}  =  n \cdot d^{k-1} \cdot n^{-(1 + \delta')}\text{  [since }k\delta = 1 + \delta'\text{]}\\
		&=   \frac{d^{k-1}}{n^{\delta'}}\text{, which is low probability for a fixed positive constant }k\text{.}
	\end{align*}
\end{proof}

\begin{lemma}\label{lemma-no-chain-of-length-k-even-in-pretend}
	The following statement holds with high probability: $\forall v \in \texttt{Honest}$, the Byzantine nodes cannot make $v$ believe that there is a chain of length $\geq k$ composed entirely of 
	Byzantine nodes, without shutting $v$ down.
\end{lemma}
\begin{proof}
	Consider an honest node $v$. If $v$ has no Byzantine neighbors in $G$, then $v$ gets true neighborhood information from all its neghbors in $G$, and is thus able to accurately reconstruct the 
	exact topology of its $k$-distance neighborhood in $H$ (please refer to Observation \ref{obs-distinguish-between-H-and-L}). Since $H$ does not have any $k$-length chain of Byzantine nodes (please 
	see Observation \ref{no-Byzantine-chain-of-length-k}), $v$'s reconstruction will have none either.
	
	So suppose $v$ has one or more Byzantine neighbors in $G$.
	Let $\mathcal{C}_1$ be the final $k$-length chain whose existence the Byzantine nodes are trying to ``trick'' $v$ into believing. Since, in truth, $\mathcal{C}_1$ has at most $k - 1$ Byzantine 
	nodes (thanks to Observation \ref{no-Byzantine-chain-of-length-k}), there must be a \emph{dummy} node $b_2$, say, which the Byzantine nodes will try to insert into $\mathcal{C}_1$ (that is, to make 
	it look as such in $v$'s eyes).
	
	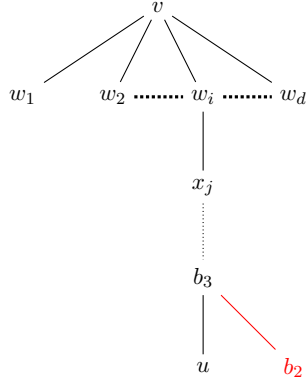
\begin{figure}
	\begin{center}
	
	\begin{tikzpicture}[grow = down, node distance=1.0cm,scale=0.8, every node/.style={scale=0.8}]
	\tikzstyle{v}=[circle,draw=black,fill=white!30,very thick,inner sep=2pt,minimum size=8mm]
	\node(v){$v$}
		child{node(w1){$w_1$}}
		child{node(w2){$w_2$}}
		child[solid]{node(wi){$w_i$}
			child[solid]{node (xj){$x_j$}
				child[densely dotted]{node (b3){$b_3$}
					child[missing]
					child[solid]{node (u){$u$}}
					child[red, solid]{node (b2){$b_2$}}			
				}		
			}
		}
		child{node(wd){$w_d$}};
	\draw[very thick, densely dotted] (w2) -- (wi);
	\draw[very thick, densely dotted] (wi) -- (wd);	
	\end{tikzpicture}
	
	\end{center}
	\caption{$\mathcal{C}_1 = (w_i, x_j, \ldots, b_3, b_2)$ is the $k$-length chain the Byzantine nodes are trying to tamper with. In reality, $b_2$ is \emph{not} a child of $b_3$ (even though $b_2$ 
	is directly connected to $v$ in the graph $G$). So $b_3$ must hide the existence of a real child $u$, say, in order to concoct the existence of the fake, Byzantine child $b_2$.}
	\end{figure}

	Thus, $b_3$, the (fake) parent of $b_2$ in the chain $\mathcal{C}_1$, must report to $v$ that it has $b_2$ as a child. While doing so, however, $b_3$ will need to suppress the existence of a (real) 
	child $u$ (which may or may not be Byzantine) because $b_3$ will need to maintain its degree $d$ in $H$ (in $v$'s eyes).
	
	But as $u$ is directly connected to $v$ in $G$, the Byzantine nodes cannot disrupt the communication between $u$ and $v$. Since $u$ knows $b_3$ to be its neighbor in $H$ (we recall that $b_3$, even 
	though Byzantine, cannot lie about its ID to $u$), the algorithm would dictate that $u$ let it be known to $v$ (regardless of whether or not $u$ is Byzantine).
	
	Therefore, $v$ will have two conflicting pieces of information: it will hear from $b_3$ that $b_3$ and $u$ are not neighbors in $H$, and $v$ will hear the exact opposite from $u$. Thus, as per the 
	algorithm, $v$ will go into \emph{crash failure}, i.e., will shut itself down (see Line \ref{shut-yourself-down} of the pseudocode in 
	Algorithm \ref{byz-alg}.
\end{proof}


\subsubsection{A high-level overview of the analysis}

In this section we show that the algorithm gives a $(\frac{b}{a})$-factor approximation of $\log{n}$ with high probability, where $a  \defeq  \frac{\delta}{10k\log{(d-1)}}$ and $b   \defeq   
\frac{4}{\log{(1 + \frac{\gamma}{d})}}$, where $\gamma$ is the edge-expansion of $\texttt{Core}$. Note that $0 < a < b < 1$. We recall that $n^{1-\delta}$ is the number of Byzantine 
nodes in the network $G$, and $d$ is the uniform degree of $H$. ($H$ is a subset of $G$. For the exact definition of $H$, please refer to Section \ref{sec:model}.)

\begin{obs}\label{b-is-large-with-byz}
	$b\log{n} \geq 2D(\texttt{Core})$, where $D(\texttt{Core})$ is the diameter of $\texttt{Core}$.
\end{obs}

\paragraph*{High-level overview of the proof}
We break our analysis up into two different stages of the algorithm. We show that the following statements hold with high probability.
\begin{enumerate}
	\item Then for $i  <  a\log{n}$, at least $(1 - \epsilon)$-fraction of the good nodes do not accept $i$ to be the right estimate of $\log{n}$, and they continue with the algorithm. The rest of 
	the nodes --- i.e., at most $\epsilon$-fraction of the good nodes --- even though they have stopped generating tokens, still continue to forward tokens generated by other nodes. 
	$0 < \epsilon \leq 1$ is a constant and we can make it arbitrarily small.
	
	\item If $i = b\log{n}$, all but $o(n)$ of the remaining active nodes accept $i$ to be the estimate of $\log{n}$ and they stop producing tokens. They however continue to forward tokens generated 
	by other (if any) nodes.
\end{enumerate}

We cannot say which way a node will decide when $a\log{n} \leq i < b\log{n}$. The above two statements are, however, sufficient to give us an approximation factor of $\frac{b}{a}   =   
\frac{40k\log{(d-1)}}{\delta\log{(1 + \frac{\gamma}{d})}}$. This gives us our main result of the paper:
\begin{theorem}\label{theorem-main-result-of-the-paper}
	Algorithm \ref{byz-alg}, with high probability, solves the Byzantine counting problem with up to $O(n^{1-\delta})$ (randomly distributed) Byzantine nodes (where $\delta > 0$ is a small fixed 
	constant that depends on $d$) and  runs in $\Theta(\log^3{n})$ rounds with the guarantee that  all but an $\epsilon$-fraction of the nodes in the network (for any arbitrary small positive constant 
	$\epsilon$) have a constant factor approximation of $\log{n}$, where $n$ is the number of nodes in the network.
\end{theorem}

The proof of the above Theorem is shown in the following Sections.


\subsubsection{When $i$ is small: In particular, when $i < a\log{n}$} \label{subsection-i-is-small-with-byz}

For the sake of the analysis in this subsection only, we will consider only \emph{Byzantine-safe nodes}, i.e., only those nodes in the set $\texttt{Byz-safe}$.

We note that while $i < a\log{n}$, no token generated by a Byzantine node reaches a Byzantine-safe node (by the very definition of a Byzantine-safe node, as defined in Definition 
\ref{definition-node-categorization}). Thus for any $v  \in  \texttt{Byz-safe}  \subset  \texttt{Safe}$, the exact same analysis of Section \ref{subsection-i-is-small} remains valid. That is, we have 
the same result, i.e., Lemma \ref{theorem-lower-end}, as in the Byzantine-free setting.
\subsubsection{When $i = \Theta(\log{n})$: In particular, when $i  =  b\log{n}$}\label{subsection-i-is-large-with-byz}

We showed in Section \ref{subsection-i-is-large} that the following statement holds with probability at least $1 - \frac{1}{n^2}$: If an honest node $v$ is still active at the beginning of this phase, 
by the end of this phase, it accepts the current value of $i$, i.e., $b\log{n}$, to be a correct estimate of $\log{n}$ and terminates.

But the aforementioned analysis in Section \ref{subsection-i-is-large} takes into account tokens generated by the honest nodes only. The Byzantine nodes can generate arbitratily high colors in any 
subphase of any phase. But we argue that they too are restricted by the structure of the network $G$. In particular, we argue that a Byzantine node can push a high-colored token (that is, a token with 
color $>    \log{(d(d-1)^{i-1})}  -  \log{\log{(d(d-1)^{i-1})}}$ in phase $i$) into the network only at the beginning of a subphase, and not at some arbitrary point in the middle of a subphase. More 
specifically we show (by exploiting the structure of the network, i.e., of the graph $G$):
\begin{lemma}\label{lemma-byzantine-nodes-cannot-delay-arbitratily}
	The following statement holds with high probability: If a \emph{core node} receives a high-colored token (generated by some Byzantine node) in round $t$ in some subphase $j$, $1 \leq j \leq 
	\alpha_i$, then $1 \leq t \leq k-1$.
\end{lemma}
\begin{proof}
	Suppose not. Suppose that there is at least one core node that receives a high-colored token in some subphase $j$, where $1 \leq j \leq \alpha_i$. Let $t \geq k$ be the earliest time-instant in 
	that subphase when a core node receives a high-colored token. That is, there is some core node $v$ that receives a high-colored token from a neighbor $b$, say, in the $t^{\text{th}}$ round. Now 
	$b$ has to pretend that it received the high-colored token from somebody else (because $b$ is not allowed to generate a token itself in the middle of a subphase). Since $v$ has edges (the edges 
	in $L$) to all the nodes in $B(b, k-1)$, $v$ can contact all those nodes directly and check the veracity of $b$'s claim. Since $t \geq k$, and since there are no Byzantine chains of length 
	$\geq k$ (Please see Observation \ref{no-Byzantine-chain-of-length-k} and Lemma \ref{lemma-no-chain-of-length-k-even-in-pretend}), there will be at least one honest node on any chain in 
	$B(b, k-1)$ who would testify against $b$.
\end{proof}

\begin{lemma}\label{lemma-core-node-stops-in-the-upper-end}
	For $v \in \texttt{Core}$, if $v$ is still active at the beginning of phase $i$ (when $i = b\log{n}$), then with high probability, by the end of this phase, it accepts the current value of $i$, 
	i.e., $b\log{n}$, to be a correct estimate of $\log{n}$ and terminates.
\end{lemma}
\begin{proof}
	Lemma \ref{lemma-byzantine-nodes-cannot-delay-arbitratily} says that the Byzantine nodes would not be able to push tokens into $\texttt{Core}$ after the first $(k-1)$-rounds of a subphase without 
	getting caught. So suppose that one or more Byzantine nodes introduce a sufficiently high color (so as to satisfy Line \ref{criterion-for-continuing} of the pseudocode) into $\texttt{Core}$ within 
	the first $(k-1)$-rounds of some subphase $j$ of the $i^{\text{th}}$ phase, where $i = b\log{n}$.

	But once even one core node receives a high color --- $\texttt{Core}$ being an expander (Please refer to Lemma \ref{lemma-core-is-an-expander}) --- that high color will start propagating through 
	the network by means of flooding and will therefore reach every uncrashed node $v$ within $D(\texttt{Core})$ rounds, where $D(\texttt{Core})$ is the diameter of $\texttt{Core}$. By Observation 
	\ref{b-is-large-with-byz}, this means that the highest color introduced by the Byzantine nodes will reach every core node $v$ within $(\frac{b\log{n}}{2} + k - 1)$-rounds.
	
	In other words, for any core node $v$, $v$ will receive no higher color in round $i$ than what it has already received before. This will violate the criterion for continuing, in particular, the 
	variable $FlagTerminate$ will \emph{not} be assigned the value $0$ (Please see Line \ref{criterion-for-continuing-with-byz} of the pseudocode in Algorithm \ref{byz-alg}. Therefore $v$ will accept 
	the current value of $i$, which is $b\log{n}$, to be a correct estimate of $\log{n}$ and will terminate.
\end{proof}

Thus we have
\begin{lemma}\label{theorem-upper-end-with-byz}
	If $i = b\log{n}$, after the $i^{\text{th}}$ phase, the following statement holds with high probability: All but $o(n)$ of the nodes that were active at the beginning of this phase accept $i$ to 
	be the correct estimate of $\log{n}$.
\end{lemma}
\begin{proof}
	Follows from Lemma \ref{lemma-core-node-stops-in-the-upper-end} and Lemma \ref{lemma-core-is-an-expander}.
\end{proof}

Lemma \ref{theorem-upper-end-with-byz} together with Lemma \ref{theorem-lower-end} give us Theorem \ref{theorem-main-result-of-the-paper}, which is the main result of this paper.


\section{Conclusion and Open Problems}

In this paper, we take a step towards  designing localized, secure, robust, and scalable algorithms for large-scale networks. We presented a fast (running in $O(\log^3{n})$ rounds) and lightweight (only simple local computations per node per round) distributed protocol for the fundamental Byzantine counting problem  tolerating  $O(n^{1 - \delta})$ (for any constant $\delta > 0$) Byzantine nodes while using only small-sized communication messages per round. Our work leaves many questions open.

A key open problem is to show a lower bound that is essentially tight with respect to the amount of Byzantine nodes that can be  tolerated, or show an algorithm that can tolerate significantly more Byzantine nodes. Our protocol works only when the Byzantine nodes are randomly distributed; it will be good to remove this assumption and design a protocol that works under Byzantine nodes that are adversarially distributed. Another interesting question is whether one can improve the approximation factor of the estimate of $\log{n}$ to $1 \pm o(1)$.


\newpage
\bibliographystyle{plain}
\bibliography{byzantine_counting_references}


\appendix

\newpage
\section{$H(n,d)$ random regular graph: definitions and properties} \label{section-locally-tree-like-property}

In this section, we formally define the $d$-regular random graph model that we are assuming and also state and prove some crucial properties that we will use in the analysis.

\subsection{Definitions}
We assume a random regular graph that is constructed by the union of $d$ random permutations as described below. Call such a random graph model, the $H(n,d)$ model (or simply {\em H-graphs}). This model was also used by Law and Siu~\cite{Law_2003} to model Peer-to-Peer networks. A random graph in this model can be constructed by picking $\frac{d}{2}$  (assume $d$ is even) Hamilton cycles independently and uniformly at random among all possible Hamilton cycles on the set of $n$ vertices, and taking the union of these Hamilton cycles. This construction yields a random $d$-regular graph (henceforth called as a $H(n,d)$ graph) that can be shown to be an expander with high probability (cf. Lemma \ref{th:friedman}). Note that a $H(n,d)$ graph is  $d$-regular multigraph  whose set of edges is composed of the $\frac{d}{2}$ Hamilton cycles. Friedman's \cite{Friedman_1991} result below  (rephrased here for our purposes) shows that a $H(n,d)$ graph is an expander (in fact, a {\em Ramanujan Expander}, i.e., the second smallest eigenvalue for these random graphs is close to the best possible) with high probability.

\begin{lemma}[\cite{Friedman_1991, Law_2003}]\label{th:friedman}
	A random $n$-node, $d$-regular $H(n,d)$-graph (say, for  $d \geq 6$)  is an expander with high probability.
\end{lemma}

\subsection{Properties}
We next show some basic properties of the $H(n,d)$ random graph which are needed in the analysis. We show some bounds on the sizes of $B(w,r)$ and $Bd(w,r)$.

\begin{lemma}
	\begin{enumerate}[(1)]
		\item $|Bd(w,r)| \leq (d-1)|Bd(w,r-1)|$.
		\item W.h.p. $|Bd(w,r)| \geq (d-1 -o(1)) |Bd(w,r-1)|$, for $1  <  r  <  \frac{\log{n}}{2\log{d}}$.
		\item For some constant $c$ and $c'$, $c'(d-1)^r \leq |B(w,r)| \leq c(d-1)^r$, w.h.p.
		\item $|B(w,r)| = \Theta(|Bd(w,r)|)$.
	\end{enumerate}
\end{lemma}
\begin{proof}
Since the degree of each node is $d$, (1) follows. From (1) it is easy to show the upper bound on $|B(w,r)|$ in (3).

We next show (2).

We first bound the expected number of neighbours that a  node $ u \in Bd(w,r-1))$ has in $B(w,r-1)$. The expected number of neighbors of  $u$  in $B(w,r)$ is  $\frac{(d-1)(n - |B(w,r-1)|)}{n}  \leq   
(d-1)(1 - \frac{\sqrt{n}}{n})$, since $|B(w,r-1)| < d^{\frac{\log{n}}{2\log{d}}} = \sqrt{n}$. Hence the expected  number of nodes in $Bd(w,r)$ is $|Bd(w,r-1)|(d-1)(1 - \frac{\sqrt{n}}{n})$. The high 
probability bound can be obtained via a Chernoff bound (one can consider the choices made by individual nodes as essentially independent if one regards the ``sampling without replacement" due to the 
permutations. This can be done if one pretends that the sample is from a set of size $n - \sqrt{n}$ (instead of $n$). This will not make a difference asymptotically.

The lower bound of (3) follows from (2) and (4) follows from (1), (2), and (3).
\end{proof}

Next we establish  the ``locally tree-like'' property of an $H(n,d)$ random graph: i.e., for most nodes $w$,  the subgraph  induced by $B(w,r)$ up to a certain radius $r$ looks ``like a tree''. This is 
stated more precisely as follows.

\begin{definition}\label{defn-typical-node--section-random-graph}
	Let $G$ be an $H(n,d)$ random graph and $w$ be any node in $G$. Consider the subgraph induced by $B(w,r)$ for $r  =  \frac{\log{n}}{10\log{d}}$. Let $u$ be any node in $Bd(w,j)$, $1 \leq j < r$. 
	$u$ is said to be ``typical'' if $u$ has only one neighbor in $Bd(w,j-1)$ and $(d-1)$-neighbors in $Bd(w,j+1)$; otherwise it is called ``atypical''.
\end{definition}

\begin{definition}\label{defn-locally-tree-like-node--section-random-graph}
	We call a node $w$ ``locally tree-like'' if no node in $B(w,r)$ is atypical. In other words, $w$ is ``locally tree-like'' if the subgraph induced by $B(w,r)$ is a $(d-1)$-ary tree.
\end{definition}

The following lemma shows that most nodes in $G$ are locally tree-like.
 
\begin{lemma}\label{lemma-most-nodes-are-locally-tree-like-section-random-graph}
In an $H(n,d)$ random graph, with high probability, at least $n - O(n^{0.8})$ nodes are locally tree-like.
\end{lemma}
\begin{proof}
Consider a node $w \in V$. We upper bound the probability that a node in $B(w,r)$, where $r  =  \frac{\log{n}}{10\log{d}}$, is atypical. For any $1 \leq j < r$,
\begin{center}
	$\Pr(u \in B(w,j)\text{ is atypical})  \leq  (d-1)\cdot\frac{|B(w,j)|}{n}  =  O(\frac{1}{n^{0.9}})$,
\end{center}
using the bound that $|B(w,j)| \leq d^r$ (the above  upper bounds the probability that $u$ has more than one neighbor in $B(w,j)$, in which case it is atypical). Hence the probability that there is some 
node $u$ that is atypical in $B(w,r)$ is $O(\frac{n^{0.1}}{n^{0.9}}) = O(\frac{1}{n^{0.8}})$. Hence the probability that node $w$ is not locally tree-like is at most $O(\frac{1}{n^{0.8}})$.

Let the indicator random variable $X_w$ indicate the event that node $w$ is locally tree-like. Let random variable $X = \sum_{w \in V}X_w$ denote the number of nodes that are locally tree-like. By 
linearity of expectation, using the above probability bound, it follows that the expected number of nodes in $G$ that are not locally tree-like is at most $O(n^{0.2})$; in other words, $E[X]   \geq   
n - O(n^{0.2})$.

To show concentration of $X$, we use Azuma's inequality (\cite{Dubhashi_2009}, Theorem $5.3$) as follows. Changing the value of $X_w$ affects only the  nodes within radius $r'  =  2r  =  
\frac{2\log{n}}{10\log{d}}$, i.e., at most $n^{0.2}$ nodes and hence affects $E[X]$ by at most $n^{0.2}$. Thus, we have
\begin{center}
	$\Pr(|X - E[X]| > n^{0.8})   \leq   2\text{exp}(-\frac{n^{\frac{8}{5}}}{n \times n^{\frac{2}{5}}})   =   2\text{exp}(-n^{\frac{1}{5}})$.
\end{center}
Hence, with high probability, at least $n - O(n^{0.8})$ nodes are locally tree-like.
\end{proof}

We now show a property that will be useful in our analysis; this follows immediately from the definition of locally-tree like and the regularity of the graph.
\begin{corollary} \label{subtrees-are-isomorphic}
Let $G$ be an $H(n,d)$ random graph and consider a node $w$ in $G$. Assume that $w$ is locally tree-like, i.e., the subgraph induced by $B(w,r)$, where $r  =  \frac{\log{n}}{10\log{d}}$ is a tree.  For 
every neighbor $u$ of $w$, the respective subtrees rooted at $u$ (in  the subgraph induced by $B(w,r)$) are isomorphic; in particular each is a $(d-1)$-ary tree.
\end{corollary}

\newpage
\section{Proof of Lemma \ref{lemma-probability-of-error-in-the-lower-end}} \label{section-proof-of-induction-analysis-in-the-lower-end}

We recall that phase $i$ consists of $\alpha_i$ subphases, and the subphases are indexed by $j$.

\begin{definition} \label{definition-failure-i-j}
    Let $Failure(i, j)$ be the event that in the $j^{\text{th}}$ subphase of the $i^{\text{th}}$ phase, $\exists t < i$ such that $k_t \geq k_i$. That is, $Failure(i, j)$ is the event that in the $j^{\text{th}}$ subphase of the $i^{\text{th}}$ phase, the node $v$ receives the maximum color in some round $t < i$.
\end{definition}

In the same vein, we define
\begin{definition} \label{definition-failure-i}
	$Failure(i)  \defeq  \bigcap_{j = 1}^{\alpha_i} Failure(i, j)$,
\end{definition}

\begin{obs} \label{obs-probability-wrong-leq-failure-i}
    We observe that the variable $FlagTerminate$ in the pseudocode (please refer to Algorithm \ref{alg}) remains $1$ after all the $\alpha_i$ subphases if the event $Failure(i)$ occurs (Please see Line \ref{criterion-for-continuing} of Algorithm \ref{alg}). In other words, a node $v$ accepts $i$ as the estimate of $\log{n}$ (and thus makes a wrong decision) if the event $Failure(i)$ occurs. Hence
    \begin{center}
        $\Pr[\text{a {\emph{safe}} node }v\text{ makes a wrong decision in the }i^{\text{th}}\text{ phase}]  \leq  \Pr[Failure(i)]$.
    \end{center}
\end{obs}

\begin{obs} \label{obs-induction-basis}
    $\Pr[Failure(1)] = 0$.
\end{obs}

\paragraph{Induction Hypothesis.} Let $i'$ be a positive integer such that $1 \leq i' < i$. Then
\begin{center}
    $\Pr[\text{a {\emph{safe}} node }v\text{ makes a wrong decision in the }{(i')}^{\text{th}}\text{phase}]   <   \frac{\epsilon}{2^{i' + 1}}$,
\end{center}
where $\epsilon$ is the error parameter.

\vspace{4 mm}

\begin{remark}
    Observation \ref{obs-induction-basis} serves as the basis of induction.
\end{remark}

\begin{lemma} \label{lemma-probability-that-inner-ring-colors-are-too-high}
    Let $E_{i, j, 1}$ be the event that $k_t  >  2(l_{i-1} - \log{(d-2)})$ for some $0 < t < i$. Then $\Pr[E_{i, j, 1}]   \leq   \frac{d-2}{d(d-1)^{i-1}}$.
\end{lemma}
\begin{proof}
    $E_{i, j, 1}$ occurs if and only if $c^{\text{max}}_{B^*(v, i-1)}  >  2(l_{i-1} - \log{(d-2)})$ (please see Line \ref{received-highest-color} and Line \ref{criterion-for-continuing-first-line} of Algorithm \ref{alg}). Thus
    \begin{align*}
        &\Pr[E_{i, j, 1}]\\
        &=   \Pr[c^{\text{max}}_{B^*(v, i-1)}  >  2(l_{i-1} - \log{(d-2)})]\\
        &\leq   \frac{d-2}{d(d-1)^{i-1}} \tag{by Lemma \ref{lemma-inner-ball-probability}}
    \end{align*}
\end{proof}

\begin{lemma} \label{lemma-probability-that-max-node-generates-a-high-enough-color-considering-everything}
    Let $E_{i, j, 2}$ be the event that $k_i  \leq  l_i - \log{(d-1)} - \log{(l_i - \log{(d-1)})}$. Then
    \begin{center}
        $\Pr[E_{i, j, 2}]   <   \frac{\epsilon}{2}  +  \frac{1}{d(d-1)^{i-1}}$.
    \end{center}
\end{lemma}
\begin{proof}
    $E_{i, j, 2}$ occurs if and only if $c^{\text{max}}_{Bd(v, i)}   \leq   l_i - \log{(d-1)} - \log{(l_i - \log{(d-1)})}$ (please see Line \ref{received-highest-color} and Line \ref{criterion-for-continuing-first-line} of Algorithm \ref{alg}).

    Let $v^{\text{max}}  \in  Bd(v, i)$ be the node that generates (or \emph{any} one of the nodes that generate) the color $c^{\text{max}}_{Bd(v, i)}$. Let $E^{\text{bad}}_{i, v^{\text{max}}}$ be the event that $v^{\text{max}}$ went inactive (i.e., took a wrong decision) in some phase $i' < i$. Then
    \begin{align*}
        &\Pr[E^{\text{bad}}_{i, v^{\text{max}}}]\\
        &=   \sum_{i' = 1}^{i - 1} \Pr[v^{\text{max}} \text{ went inactive in phase }i']\\
        &<   \sum_{i' = 1}^{i - 1} \frac{\epsilon}{2^{i' + 1}} \tag{by the induction hypothesis}\\
        &<   \sum_{i' = 1}^{\infty} \frac{\epsilon}{2^{i' + 1}}   =   \frac{\epsilon}{2}\text{.}
    \end{align*}

    That is,
    \begin{equation} \label{equation-probability-that-max-node-went-inactive-in-some-previous-phase}
        \Pr[E^{\text{bad}}_{i, v^{\text{max}}}]   <   \frac{\epsilon}{2}
    \end{equation}

    If $v^{\text{max}}$ is still active in the current phase, i.e., $v^{\text{max}}$ did \emph{not} go inactive in some previous phase $i' < i$, then in order to calculate $\Pr[E_{i, j, 2}]$, it is enough to consider $Bd(v, i)$ \emph{in its entirety} along with the properties of the geometric distribution. That is,
    \begin{align*}
        &\Pr[E_{i, j, 2}\ |\ (E^{\text{bad}}_{i, v^{\text{max}}})^c]\\
        &=   \Pr[c^{\text{max}}_{Bd(v, i)}   \leq   l_i - \log{(d-1)} - \log{(l_i - \log{(d-1)})}]]\\
        &<   \frac{1}{d(d-1)^{i-1}} \tag{by Lemma \ref{lemma-ball-boundary-probability}}\\
    \end{align*}

    That is,
    \begin{equation} \label{equation-probability-that-max-node-generates-a-high-enough-color}
        \Pr[E_{i, j, 2}\ |\ (E^{\text{bad}}_{i, v^{\text{max}}})^c]   <   \frac{1}{d(d-1)^{i-1}}
    \end{equation}

    Combining Equations \ref{equation-probability-that-max-node-went-inactive-in-some-previous-phase} and \ref{equation-probability-that-max-node-generates-a-high-enough-color}, we get that
    \begin{align*}
        &\Pr[E_{i, j, 2}]\\
        &\leq   \Pr[E^{\text{bad}}_{i, v^{\text{max}}}]  +  \Pr[E_{i, j, 2}\ |\ (E^{\text{bad}}_{i, v^{\text{max}}})^c] \tag{thanks to Fact \ref{fact-standard-probabilistic-inequality}} \\    
        &<   \frac{\epsilon}{2}  +  \frac{1}{d(d-1)^{i-1}}\text{.}
    \end{align*}
\end{proof}

\begin{lemma} \label{lemma-probability-that-the-max-color-comes-from-the-outermost-ring-considering-everything}
    Let $Success(i, j)$ be the event that
    \begin{enumerate}
        \item the maximum color received by node $v$ until the $(i-1)^{\text{th}}$ round of the $j^{\text{th}}$ subphase of the $i^{\text{th}}$ phase is strictly less than the maximum color received by node $v$ in the $i^{\text{th}}$ round of the same subphase of the same phase. That is, in terms of the pseudocode (please refer to Line \ref{received-highest-color} and Line \ref{criterion-for-continuing-first-line} of the pseudocode), $k_t  <  k_i$, $\forall t < i$.
	
        And
	
        \item $k_i  >  l_i - \log{(d-1)} - \log{(l_i - \log{(d-1)})}$.
    \end{enumerate}

    Then $\Pr[Success(i, j)]   >   1  -  (\frac{1}{d(d-1)^{i-2}}  +  \frac{\epsilon}{2})$.
\end{lemma}
\begin{proof}
    One of the ways $Success(i, j)$ can happen is if $k_t   \leq   2(l_{i-1} - \log{(d-2)})$, $\forall t < i$, \emph{and} $k_i  >  l_i - \log{(d-1)} - \log{(l_i - \log{(d-1)})}$. Thus
    \begin{align*}
        &Success(i, j)  \supset  (E_{i, j, 1}^c  \cap  E_{i, j, 2}^c)\\
        &\implies  \Pr[Success(i, j)]  \geq  \Pr[E_{i, j, 1}^c  \cap  E_{i, j, 2}^c]\\
        &=   \Pr[(E_{i, j, 1} \cup E_{i, j, 2})^c] \tag{by De Morgan's Theorem}\\
        &=   1  -  \Pr[E_{i, j, 1} \cup E_{i, j, 2}]\\
        &\geq   1 - \Pr[E_{i, j, 1}] - \Pr[E_{i, j, 2}]\\
        &\tag{since, by the union bound, $\Pr[E_{i, j, 1} \cup E_{i, j, 2}]   \leq   \Pr[E_{i, j, 1}] + \Pr[E_{i, j, 2}]$}\\
        &>   1 - \frac{d-2}{d(d-1)^{i-1}} - \frac{1}{d(d-1)^{i-1}} - \frac{\epsilon}{2} \tag{From Lemma \ref{lemma-probability-that-inner-ring-colors-are-too-high} and Lemma \ref{lemma-probability-that-max-node-generates-a-high-enough-color-considering-everything}}\\
        &=   1  -  (\frac{1}{d(d-1)^{i-2}}  +  \frac{\epsilon}{2})\text{.}
    \end{align*}
\end{proof}

\begin{lemma}\label{lemma-lower-end-probability-of-failure-in-phase-i-and-subphase-j}
	$\Pr[Failure(i, j)]   <   \frac{1}{d(d-1)^{i-2}}  +  \frac{\epsilon}{2}$.
\end{lemma}
\begin{proof}
    We observe that $Failure (i, j)   =  (Success(i, j))^c$ and the result immediately follows from Lemma \ref{lemma-probability-that-the-max-color-comes-from-the-outermost-ring-considering-everything}.
\end{proof}

\begin{lemma}
	$\Pr[\text{a {\emph{safe}} node }v\text{ makes a wrong decision in the }i^{\text{th}}\text{phase}]   <   \frac{\epsilon}{2^{i+1}}$.
\end{lemma}
\begin{proof}
\begin{align*}
	&\Pr[Failure(i)]  =  \prod_{j = 1}^{\alpha_i} \Pr[Failure(i, j)]\\
	&\text{  [since the subphases are independent from each other]}\\
	&<   \prod_{j = 1}^{\alpha_i}  \frac{1}{d(d-1)^{i-2}}\text{  [since }\Pr[Failure(i, j)]  <  \frac{1}{d(d-1)^{i-2}}\text{from Lemma 
	\ref{lemma-lower-end-probability-of-failure-in-phase-i-and-subphase-j}]}\\
	&=   (\frac{1}{d(d-1)^{i-2}})^{\alpha_i}
\end{align*}
If we set $\alpha_i  \defeq  \lceil\frac{\log{(\frac{1}{\epsilon})} + i + 1 - \log{d}}{(i-2)\log{(d-1)}}\rceil$, then
\begin{align*}
	&\alpha_i  \geq  \frac{\log{(\frac{1}{\epsilon})} + i + 1 - \log{d}}{(i-2)\log{(d-1)}}\\
	&\implies   (\frac{1}{d(d-1)^{i-2}})^{\alpha_i}  \leq   \frac{\epsilon}{2^{i+1}}\\
	&\implies   \Pr[Failure(i)]   <   (\frac{1}{d(d-1)^{i-2}})^{\alpha_i}  \leq   \frac{\epsilon}{2^{i+1}}
\end{align*}

Thanks to Observation \ref{obs-probability-wrong-leq-failure-i},
\begin{center}
	$\Pr[\text{a {\emph{safe}} node }v\text{ makes a wrong decision in the }i^{\text{th}}\text{ phase}]  \leq  \Pr[Failure(i)]  <  \frac{\epsilon}{2^{i+1}}$.
\end{center}
\end{proof}


\end{document}